\documentclass{article}
\usepackage[dvipsnames]{xcolor}
\usepackage{subcaption}
\usepackage{booktabs}
\usepackage{caption}
\usepackage{graphicx,enumitem,placeins}
\usepackage{float}
\usepackage{array}
\usepackage{amsmath}
\numberwithin{equation}{section}
\usepackage{hyperref}
\usepackage[margin=1in]{geometry}
\usepackage{amssymb}
\usepackage{amsthm}
\theoremstyle{plain}             
\newtheorem{theorem}{Theorem}[section]
\newtheorem{lemma}[theorem]{Lemma}
\newtheorem{proposition}[theorem]{Proposition}
\renewcommand{\thesection}{\Alph{section}}
\renewcommand{\thesubsection}{\Alph{section}.\arabic{subsection}}

\title{A High-Level Framework for Practically Model-Independent Pricing}

\author{Marco Airoldi \\ Mediobanca, Piazzetta Enrico Cuccia 1, Italy \\ Marco.Airoldi@mediobanca.com}
\date{2025-12-02}

\begin{document}
\maketitle
\begin{abstract}
We present a high-level framework that explains why, in practice, 
different pricing models calibrated to the same vanilla surface 
tend to produce similar valuations for exotic derivatives. 
Our approach acts as an overlay on the Monte Carlo infrastructure 
already used in banks, combining path reweighting with a 
conic optimisation layer without requiring any changes to 
existing code. This construction delivers narrow, 
practically model-independent price bands for exotics, reconciling front-office 
practice with the robust, model-independent ideas developed in the academic 
literature.
\end{abstract}
%
\section*{Introduction}
In incomplete markets, the same exotic pay-off may yield markedly 
different prices under different forward-volatility assumptions, 
even if all models match the observed spot-volatility surface exactly.
No-arbitrage therefore fails to single out a unique fair value, so exotic 
derivatives must be priced under genuine model uncertainty.

\cite{Avellaneda1996_a,Avellaneda1996_b} pioneered robust hedging via
worst-case volatility scenarios, enforcing no-arbitrage constraints to
derive model-independent bounds.  \cite{Hobson1998} recast the problem
as a Martingale Optimal Transport (MOT) programme, subsequently
formalised by \cite{Beiglbock2013} and \cite{Labordere2014}.
Although theoretically elegant, MOT still faces scalability issues in
high-dimensional or path-dependent settings \cite{DolinskySoner2014}.
Monte Carlo techniques under model uncertainty--see 
\cite{Brigo2015}--underscore similar computational 
and dimensionality challenges for robust pricing. 
Within this paradigm, Avellaneda et al.\ \cite{Avellaneda2001} introduced the weighted Monte Carlo 
methodology, reweighting simulated paths to match observed vanilla smiles while staying as close as 
possible to a chosen prior measure. \\
Foundational results on static volatility surfaces due to
\cite{Dupire1994} and \cite{Breeden1978} clarified how option prices
encode marginal distributions, but left open the dynamics of forward
volatility.  Extensions to forward-start options by
\cite{Carmona2009} and \cite{Carr2010} highlighted further calibration
challenges.  Even sophisticated exotic-local-volatility constructions
\cite{Guyon2012} cannot eliminate the residual model risk emphasised 
by~\cite{Guo2018}. 

Entropic MOT. \;A recent line of work introduces entropic
regularisation to MOT, achieving faster convergence and enhanced
numerical stability through Sinkhorn-type iterations
\cite{Chen2024,DeMarch2018,Tang2025}.  These methods successfully
blend robustness with computational tractability, paving the way for
real-time applications in quantitative finance. \\
Relation to Schr\"odinger bridges and prior-invariance. Entropic 
MOT / Schr\"odinger bridge methods select, among all martingale 
measures matching vanilla data, the one closest to a reference 
prior $P_0$ in Kullback-Leibler divergence; the prior is thus 
ontologically privileged and the solution remains a ``bridge'' 
around it. In contrast, our augmented fixed-point convex 
program is constraint-centric and generator-invariant: once 
exact vanilla replication and shape constraints are enforced 
and the double variance penalty is active, the unique 
minimiser ($w^{\*}$ ,$\sigma^{\*}$) does not depend on the 
starting Monte Carlo generator, as long as the supplied 
paths span the same linear space of pay-offs. Formally, 
strict convexity yields uniqueness of the fixed-point limit, 
while projector equivalence of pay-off matrices implies 
equality of feasible sets and optima; hence any ``house''
simulator becomes exchangeable, and the result is prior-independent 
(see Theorem~\ref{thm:aug_uniqueness} and Proposition~\ref{prop:span_equiv}).
Intuitively, as we range over the multiverse of vanilla-consistent 
models--one for each weight vector $w$-the observable constraints 
compel all admissible dynamics to tell the same market story; 
variance penalties then act as an Occam selector that 
collapses these equivalent narratives into a single, 
generator-invariant representative.

\subsection*{Our Key Contribution}

We present a Monte Carlo reweighting framework that finally reconciles 
two extremes of option pricing: the concreteness of a calibrated model 
and the rigour of model--independent / robust pricing envelopes.  A single LP step lifts 
any ``house'' simulation into a multiverse of weighted paths, spanning 
robust limits yet collapsing back to the reference model when warranted.  
This framework turns the long--sought oxymoron of a ``model--independent model'' into 
practical reality, in a framework engineered for seamless front-office use.

Key capabilities:

\begin{enumerate}
  \item \textbf{Computational Scalability} \\
        Handles a Reverse Cliquet with up to 100 fixing dates and roughly
        1\,000 linear constraints in a couple of hours on an end-user
        CPU, using open source conic solvers.

  \item \textbf{Plug-and-Play}\\
        Runs on the institution's existing Monte Carlo paths, wrapping any
        legacy engine (e.g. Heston, Merton, constant volatility) in an
        "exoskeleton" that enforces vanilla consistency, no arbitrage and
        variance penalties without changing the underlying risk library or
        workflow (see Figure~\ref{smart_monte_carlo_workflow}).

  \item \textbf{Modelling model-independent prices }\\
        Our framework re--imagines model--independent pricing: rather than stopping
        at theoretical bounds, it shows that fully calibrated models naturally
        emerge from a constrained inversion. In doing so, it embeds model risk
        directly in the trading-desk workflow, justifying concrete models
        (e.g.\ Heston) while upgrading them into robust, computationally efficient 
        data--driven front--office pricing tools. 

  \item \textbf{Realistic Min--Max Bounds \& Beyond}\\
        Delivers a robust pricing envelope that preserves the forward-volatility skew and
        attains exact calibration to plain vanilla options via minimal
        weight adjustments.

\end{enumerate}

\subsection*{Reading structure.}

The paper is organized in two complementary layers of exposition. 
Sections~\ref{Section_methodology}, \ref{Section_B_inversion},~\ref{section_C_performance} 
form a self-contained narrative: \\
Sec.~\ref{Section_methodology} introduces the methodology, including a schematic workflow; \\
Sec.~\ref{Section_B_inversion} develops the advanced aspects; \\
Sec.~\ref{section_C_performance} presents the numerical results. \\
The Appendix follows the same logical progression but expands each step with full mathematical detail, 
algorithmic specifications, and auxiliary derivations. 

A glossary of all terms is provided in Appendix~\ref{app_glossary}.

\section{The methodology}
\label{Section_methodology}

\subsection{Flowchart Overview and Methodology Definition}
\textbf{The goal:}

This sub~Section introduces a step-by-step, flowchart-based pipeline 
(Figure~\ref{smart_monte_carlo_workflow}) designed for two main \textbf{use cases}: (i) reconstructing 
forward-start volatility surfaces and (ii) pricing path-dependent exotic 
derivatives. The overarching goal is to derive model-independent price ranges 
for exotics using only vanilla quotes, correctly reproducing the forward-skew 
shape while remaining fully consistent with the observed spot-volatility surface. 
Each stage of the pipeline is detailed in the next sub~Sections.
\begin{figure}[h]
    \centering
    \includegraphics[width=0.4\textwidth]{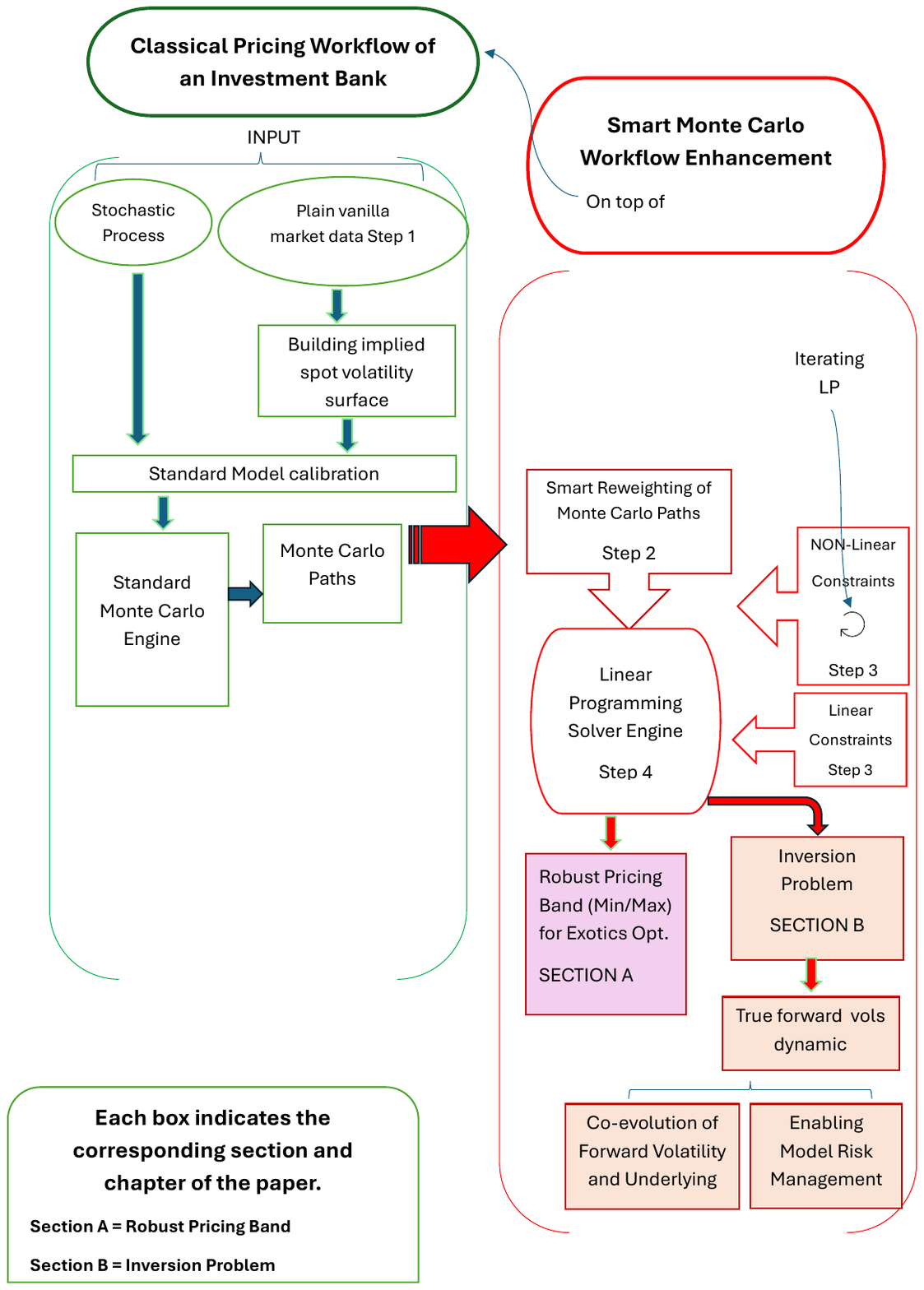}
    \caption{Smart Monte Carlo Workflow and Its Interaction with the Classical Workflow of an Investment Bank.
            The light-red blocks represent the standard front-office pricing 
            loop; the orange overlay shows the Smart Monte Carlo 
            reweighting step that enforces vanilla consistency and 
            robust min-max bounds without modifying the legacy 
            risk library.}
    \label{smart_monte_carlo_workflow}
\end{figure}

\subsection*{Step 1: Plain Vanilla Market Data (INPUT)}

The primary input for our framework consists of plain--vanilla option prices 
observed from the market. In practice we ingest these data in a single, 
fully explicit format:

The market data inputs can originate from:

\begin{enumerate}
    \item \textbf{Real market data}: Coming from plain vanilla markets.

    \item \textbf{Synthetic parametric data}: Generated by assigning specific 
    parameter values in a parametric formula.

    \item \textbf{Synthetic model-based data}: Artificially constructed using 
    a chosen set of calibrated parameters within a stochastic model.
\end{enumerate}

\subsection*{Step 2: Monte Carlo Path Generation with Variable Weights}
The next step generates Monte Carlo paths under a chosen stochastic 
model, such as the Merton jump-diffusion or Heston stochastic-volatility 
model.

We consider two main distinct calibration scenarios:

\begin{itemize}
    \item Natively calibrated scenario: Model parameters are predefined and 
    assumed correct without calibration to market data. Thus,
    plain vanilla quotes are replicated by construction, i.e.:
    \begin{equation*} \sum_{i=1}^{N}(w_i - \frac{1}{N}) \cdot \text{Payoff}_{i}(K,T) = 0 \end{equation*}
    highlighting that uniform weights $w_i = \frac{1}{N}$ constitute the most natural 
    and parsimonious solution, although other solutions might also 
    exist.

    \item Market-calibrated scenario: Model parameters are calibrated to 
    market-observed vanilla prices. Due to inevitable calibration 
    imperfections, uniform weights generally will not replicate market 
    prices exactly. 
\end{itemize}

\subsection*{Step 3: Constraints }

\begin{itemize}
\item {\bf Linear constraints}
\begin{itemize}
  \item \textbf{Non-negative path weights:} the path weights $w_i$, satisfying
    \begin{equation*}
      w_i \ge 0, \quad \sum_i w_i = 1
    \end{equation*}
    are the primary decision variables to be determined.

    \item \textbf{Vanilla and digital replication:} 
    For each maturity we impose exact replication of both plain-vanilla options 
    and digital options (implemented as tight call spreads on the same strike grid):
    \begin{equation*}
      \sum_{i} w_i\,\mathrm{Payoff}^{\mathrm{van}}_i(K,T)
      \;=\;
      P^{\mathrm{van}}_{\mathrm{market}}(K,T),
      \qquad
      \sum_{i} w_i\,\mathrm{Payoff}^{\mathrm{dig}}_i(K_d,T)
      \;=\;
      P^{\mathrm{dig}}_{\mathrm{market}}(K_d,T).
    \end{equation*}
    Matching digital quotes fixes the local slope of the implied-volatility smile 
    at the selected strikes, so that the calibration reproduces not only the level 
    of the spot-volatility skew but also its first derivative with respect to strike.

  \item \textbf{No-arbitrage:} enforce monotonicity in strike and convexity of butterfly spreads.  

\end{itemize}
\label{list_linear_constraints}

\item {\bf Non-linear constraints (Affine linearization)}

Non-linear constraints are indispensable if we want to extend 
the method's scope to include forward volatilities, because 
these volatilities are inherently non-linear: the mapping that 
links the option price-expressible as a linear combination 
of $\{w_i\}$ to volatility is itself non-linear.

Incorporating forward implied volatilities 
therefore enables us to model:
\begin{itemize}
  \item \textbf{The forward surface reconstruction}
  \item \textbf{Total-variance conservation:}
  In practice we anchor the total variance at each maturity using the 
  Carr-Madan variance-swap identity, rather than the static 
  forward-variance formula. 
  (See Appendix~\ref{app:formulation_and_implementation}.)
\end{itemize}

\end{itemize}
All of this will be laid out in Step~4. Keep in mind that this step is needed only if we want to add the 
non-linear conditions discussed above; if we limit ourselves to option prices that can already be written 
as a linear combination of pay-off weights, Step~4 can be omitted.

\subsection*{Step 4: Fixed-point scheme for forward volatility}

Up to Step 3 all constraints are linear in the scenario-weights $w_i$.
As soon as we introduce forward-start volatilities $\sigma_{i,j}$, a 
nonlinear relation appears: for each fixing date $t_i$ and strike $K_j$ 
the Black-Scholes formula gives a nonlinear map
\[
  \sigma_{i,j}
  \;\longmapsto\;
  P^{\mathrm{Fwd}}_{i,j}(\sigma_{i,j}),
\]
and the forward-start prices $P^{\mathrm{Fwd}}_{i,j}$ themselves depend 
linearly on the weights $w$. The composition produces a nonlinear 
constraint linking $(w,\sigma)$.

To keep the optimisation convex, we do not impose the full Black-Scholes 
relation inside the program. Instead, at each outer iteration we fix a 
current guess $\sigma^{\text{prev}}_{i,j}$ and linearise the map 
$\sigma \mapsto P^{\mathrm{Fwd}}_{i,j}(\sigma)$ around this point. 
For each $(t_i,K_j)$ we obtain an affine approximation that links 
$\sigma_{i,j}$ to the reweighted forward price $P_{i,j}(w)$. This 
affine relation can be enforced as a linear equality constraint inside 
an LP or SOCP.

The resulting algorithm has the following structure.

\begin{enumerate}[leftmargin=1.6em]
  \item \textbf{Initialisation.}
        Start from the Monte Carlo generator used to produce the paths 
        and uniform weights $w_i \equiv 1/N$. Compute the corresponding 
        forward-start prices $P^{\text{prev}}_{i,j}$, implied forward 
        volatilities $\sigma^{\text{prev}}_{i,j}$, and Vegas.

  \item \textbf{Inner convex problem.}
        Solve a convex program in the variables $(w,\sigma)$ that 
        includes all linear constraints from Step 3, plus the affine 
        constraints obtained from the Taylor expansion of the 
        Black-Scholes price in $\sigma$. Optional quadratic penalties 
        can be added on the dispersion of the weights and on the 
        dispersion of the forward-volatility slice.

  \item \textbf{Fixed-point update.}
        After solving the convex problem, update the expansion point by 
        setting $\sigma^{\text{prev}} \leftarrow \sigma$ and recompute 
        the forward-start prices and Vegas under the new weights. 
        Repeat the inner optimisation until the change in the 
        forward-start prices or in the implied forward volatilities is 
        below a prescribed tolerance.
\end{enumerate}

In this way the nonlinearity between prices and forward volatilities is 
handled by a short fixed-point loop, while each inner step remains a 
standard convex optimisation problem. At convergence, the forward 
volatilities that appear in the convex constraints are consistent, up to 
numerical tolerance, with the Black-Scholes implied volatilities 
associated with the reweighted forward-start prices.

\subsection*{Step 5: Min-Max optimisation of exotic option prices and resulting weights}
Once the optimal Monte-Carlo path weights $\{w_i\}$ have been 
determined, the price of any derivative pay-off is

\begin{equation*} D\;=\;\sum_{i} w_i\,F_i, \end{equation*}

where $F_i$ denotes the discounted pay-off along path $i$. We next 
minimise and maximise this linear functional to obtain tight lower 
and upper limits for the exotic instrument under consideration. 

This single optimisation step therefore delivers both the interval 
$[D_{\min},D_{\max}]$ and the corresponding extremal weight sets.

\subsection{Use Cases of the Method}
We highlight three primary applications that demonstrate the flexibility 
and robustness of our framework:

\subsubsection{Forward Volatility Surface Reconstruction}

    Recall that $\sigma_{\text{Fwd-static}}(T1,T2)$ is the model-independent 
    variance-additive forward-volatility (def. appendix A.1), whereas 
    $\sigma_{\text{Fwd-model}}$ is obtained by repricing the forward-start 
    under the candidate dynamics.

    The resulting linear program solution yields a model-independent 
    "envelope" (upper and lower limits) for the forward volatility 
    surface. Typically, we generate two opposite forward volatility 
    surfaces---one maximizing and the other minimizing either the 
    curvature of the ATM volatility surface or the value of the 
    forward-start ATM call option---in a model-independent fashion. 
    Figure~\ref{fig_forward_surfaces_implied_heston_vs_merton} 
    illustrates these cases for a 
    Reverse Cliquet pay-off (see~\ref{gloss:reverse_cliquet}); the number of fixing dates
    can be scaled up to 100 dates.

\subsubsection{Pricing Exotic Derivatives}

\begin{itemize}
    \item Once the optimal path weights are determined, we can price any 
    exotic pay-off by taking the appropriately weighted average of its 
    discounted pay-off. 

\end{itemize}

\subsubsection{Consistency and limiting-case checks.}
To validate the internal consistency of the proposed methodology, we perform two benchmark tests
that verify its behaviour in well-understood limiting cases.
(i)~As the start delay of a forward-start exotic tends to zero, the computed minimum and maximum prices
converge to those of the corresponding spot option, confirming the correct transition to the standard spot case.
See Figure~\ref{fig_fw_vol_surface_convergence_to_spot}.
(ii)~For reverse cliquet structures, increasing the number of fixing dates leads to a clear convergence
towards a constant limit value, in line with theoretical expectations~\cite{Hobson1998, Cox2011}.
These checks demonstrate that the framework consistently reproduces known asymptotic limits
and provides a reliable baseline for the numerical experiments discussed in Section~C.

%
\begin{figure}[h]
    \centering
    \includegraphics[width=0.6\textwidth]{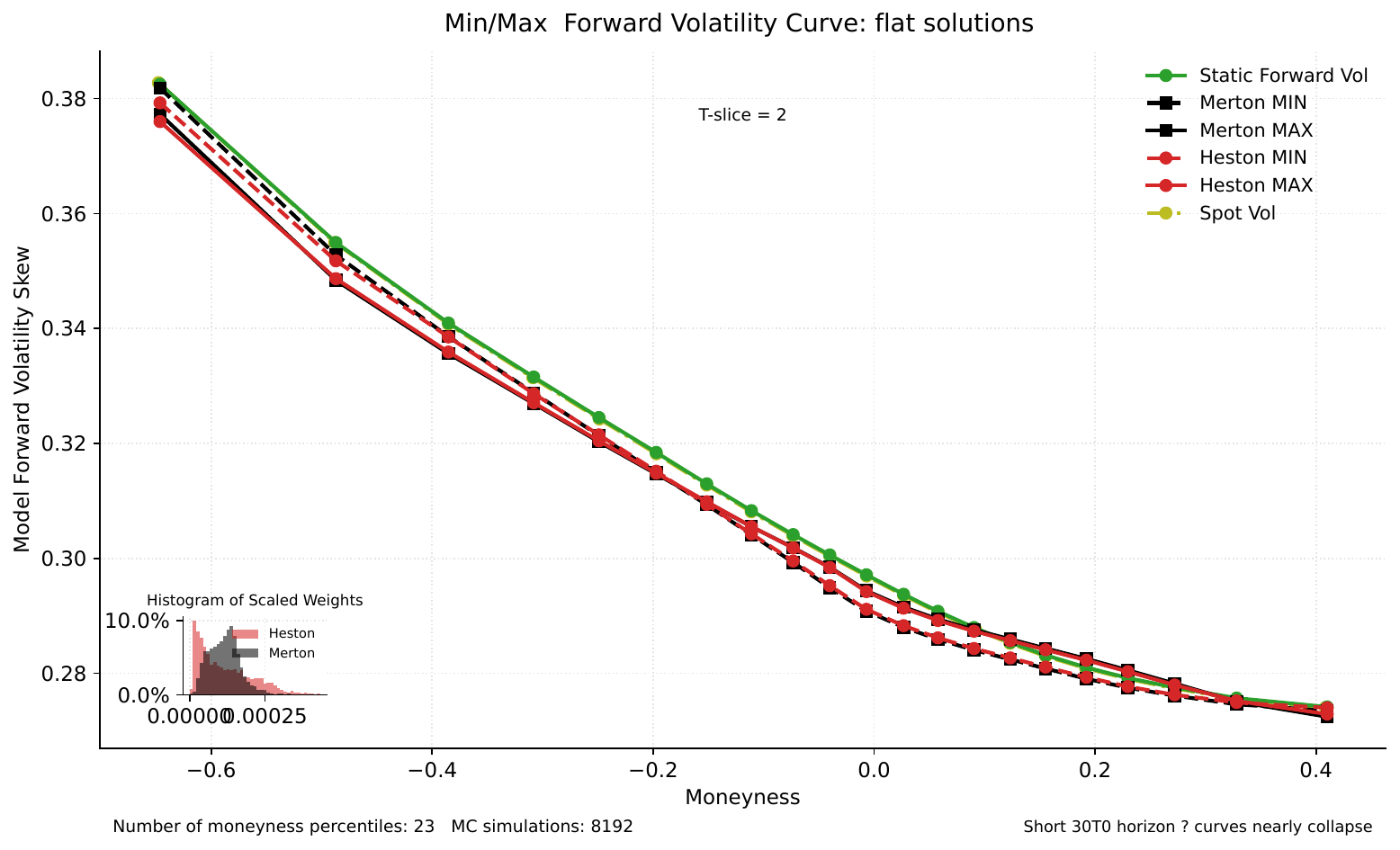}
    \caption{Forward-volatility surface converging to the spot-volatility 
     surface as the forward start date $T_1$ tends to zero.}
    \label{fig_fw_vol_surface_convergence_to_spot}
\end{figure}
%
At very short maturities the framework naturally enforces a tight connection 
between the spot implied volatility surface and the reconstructed forward surface. 
When the grid of fixing dates includes a dense sequence of points close to 
today (for example, daily nodes over the first few weeks or months), each short-dated 
slice acts as an additional linear constraint on the forward-start prices. 
In the limit where the start date of a forward-start option tends to zero, the 
pay-off degenerates into a standard spot call and its implied volatility must 
coincide with the spot slice. By construction, the optimisation inherits 
this behaviour: as $T_1 \to 0$, the model forward volatility 
$\sigma_{\mathrm{Fwd\text{-}model}}(K;T_1,T_2)$ is forced to merge with the 
spot smile $\sigma_{\mathrm{spot}}(K;T_2)$, up to numerical tolerance.

This short-maturity anchoring has an important practical consequence. Instead of relying 
only on a sparse set of maturities (for example 3M, 6M, 1Y), the Smart Monte Carlo 
calibration can ingest a rich term structure of short-dated smiles and 
propagate their information into the whole forward surface. The collective effect 
of many closely spaced spot slices severely reduces the degrees of freedom of the 
forward-start component. 
In particular, exotics whose accrual starts near today are priced under a 
forward surface that is not only calibrated to vanillas at each maturity, but also 
continuously glued to the observed spot surface in the $T_1 \to 0$ limit.

By avoiding strong parametric assumptions and direct model calibration, 
the proposed method offers a powerful, robust, and easy to implement 
framework for exotic derivative pricing. The feasible region of our 
linear optimisation corresponds to all ``admissible'' forward--volatility 
structures, while its extreme points represent scenarios that either 
maximize or minimize each exotic pay-off.

A more detailed performance analysis is presented in 
Section~\ref{section_C_performance}.


\section{From the Inversion Problem to an Independent Model }
\label{Section_B_inversion}
\setcounter{subsection}{0}
\addtocounter{subsection}{-1}
\renewcommand{\thesubsection}{\thesection.\arabic{subsection}}


\subsection*{Executive Summary}

Section~\ref{Section_B_inversion} presents the core theoretical structure of the framework and explains
the numerical phenomena observed in the motivating example. The material can
be read at three different levels.

(1) Level 1: numerical behaviour and motivating example.
The sub~Section opens with a detailed stress test based on the Reverse Cliquet,
one of the most challenging pay-offs in terms of model risk and sensitivity
to forward-volatility skew. The pay-off aggregates information across many
fixing dates, and we consider an extreme configuration with up to 100 fixing
dates and highly stressed generators (including large jump amplitudes and
high stochastic-volatility regimes). Despite these adverse conditions, the
three optimisation regimes (raw min-max, minimum dispersion of forward-vol
increments, and joint minimum dispersion of weights and forward-vol) show a
remarkable contraction of the admissible price interval. Differences between
distinct generators collapse from model-scale deviations to a few cents of
price. This behaviour provides the empirical motivation for the theory.

(2) Level 2: inversion, enlargement, and structural uniqueness.
The first part of the theory analyses the inversion problem on a fixed path
grid. Vanilla replication determines a convex set of admissible weight vectors.
A unique distribution is selected by minimising weight dispersion, which
identifies the calibrated measure closest to the ideal barycentric configuration.
When latent coordinates such as forward-start volatilities are 
introduced, the space of admissible models becomes under-determined 
again unless a new dispersion penalty is added.
This explains why additional variance blocks (for example
on the forward-volatility slice) are required to recover uniqueness. The enlarged
framework remains convex and selects a single representative model inside the
expanded feasible set.

(3) Level 3: technical layer and model set equivalence.
At the most granular level we examine the internal structure of the optimised
weights and how different path sets can be compared. Rather than enforcing a
fully isotropic Monte Carlo ensemble, we rely on simple diagnostics to prevent
degeneracy and to test whether two generators span the same pay-off space. In
practice this takes the form of a mass-splitting test for path-set equivalence
and an interior barrier that keeps weights in the interior of the simplex.
Together these tools provide the technical backbone that guarantees the
stability and robustness of the extremal valuations developed 
in subsection~\ref{sub_section_B_reverse_cliquet}.

The sub~Section concludes with a hierarchical view of expressive power. Each time
a new block of latent variables is added (for example $\sigma$, then $\gamma$), the
descriptive language of the framework expands. This enlargement increases the
class of pay-offs that can be represented, but does not remove incompleteness.
The analogy with G\"odel's incompleteness theorem captures this structure: the
constraints and variance penalties act as axioms, and the admissible calibrated
configurations correspond to the theorems derivable from them. Enriching the
model extends the language but cannot exhaust it, a property that becomes
visible in the numerical behaviour of the Reverse Cliquet.

\subsection{Motivating Example: Reverse Cliquet Stress Test}
\label{sub_section_B_reverse_cliquet}

The Reverse Cliquet is one of the most demanding non-linear pay-offs for any 
model-independent approach. It aggregates forward return increments across 
many fixing dates, is extremely sensitive to forward-volatility skew, and its natural 
bid-ask width tends to shrink as the number of fixings increases. This makes 
it an ideal stress test for weighted Monte Carlo reweighting.

We deliberately select three generators with fundamentally different 
structures:

\begin{itemize}
  \item \textbf{Merton jump diffusion}, exhibiting discontinuous paths and 
        jump driven forward-volatility skew.
  \item \textbf{Heston stochastic volatility}, producing smooth but highly 
        non-linear volatility dynamics.
  \item \textbf{Pure Black model}, a one parameter baseline 
        with no ability to generate skew, used as the weakest possible prior.
\end{itemize}

Despite their structural differences, the reweighting scheme forces all 
three priors to converge towards tightly constrained valuations under the 
three optimisation regimes described below. To interpret the numerical 
results in trading desk language, we decompose price uncertainty into an 
intra model bid-ask component and a model risk component, and we express 
both in relative form.
%
%
\FloatBarrier
\begin{figure}[t]
  \centering
  \begin{subfigure}[t]{0.45\linewidth}
    \centering
    \includegraphics[width=\linewidth]{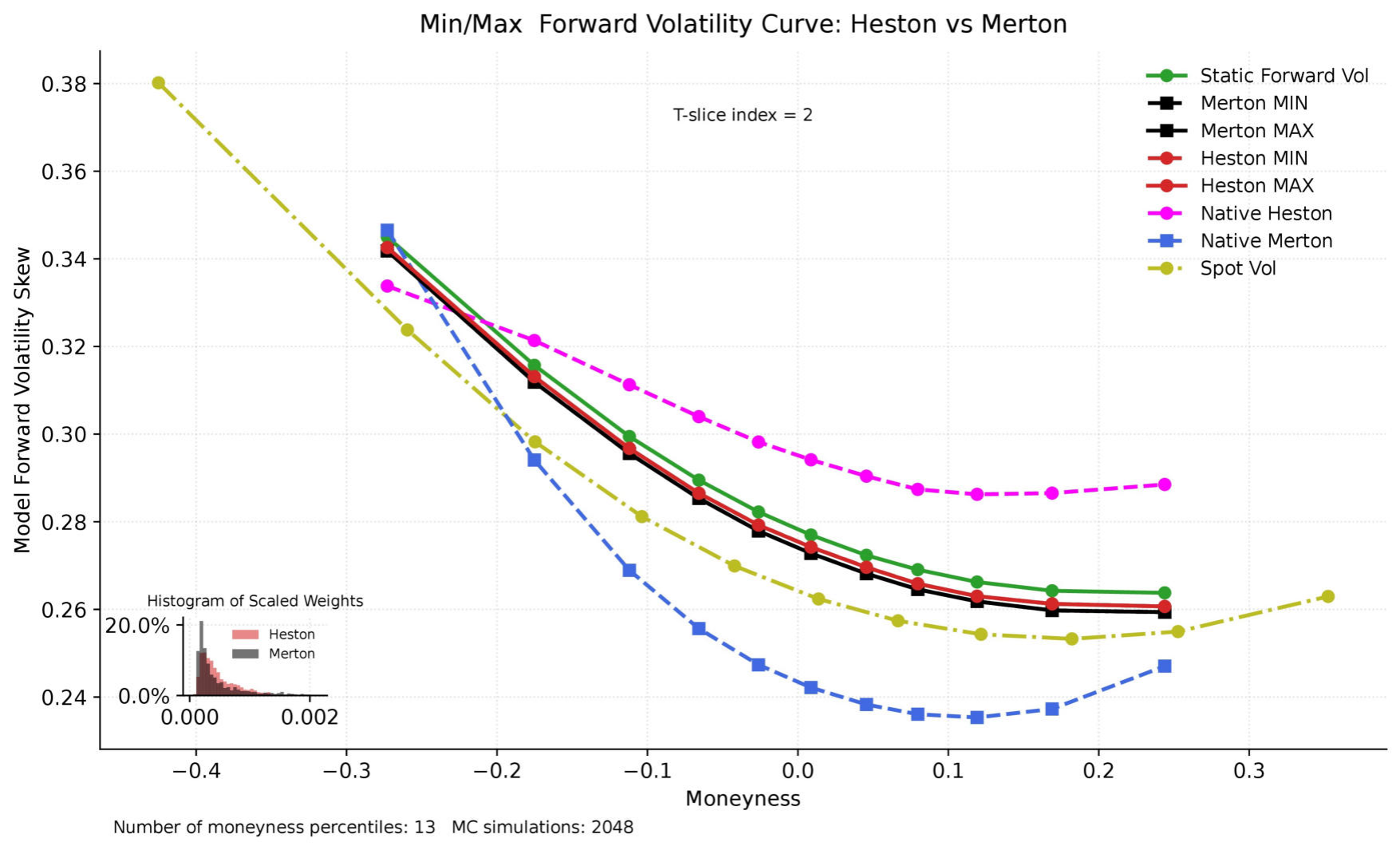}
    \caption{Heston vs Merton}
    \label{fig_forward_surfaces_implied_heston_vs_merton}
  \end{subfigure}
  \hfill
  \begin{subfigure}[t]{0.45\linewidth}
    \centering
    \includegraphics[width=\linewidth]{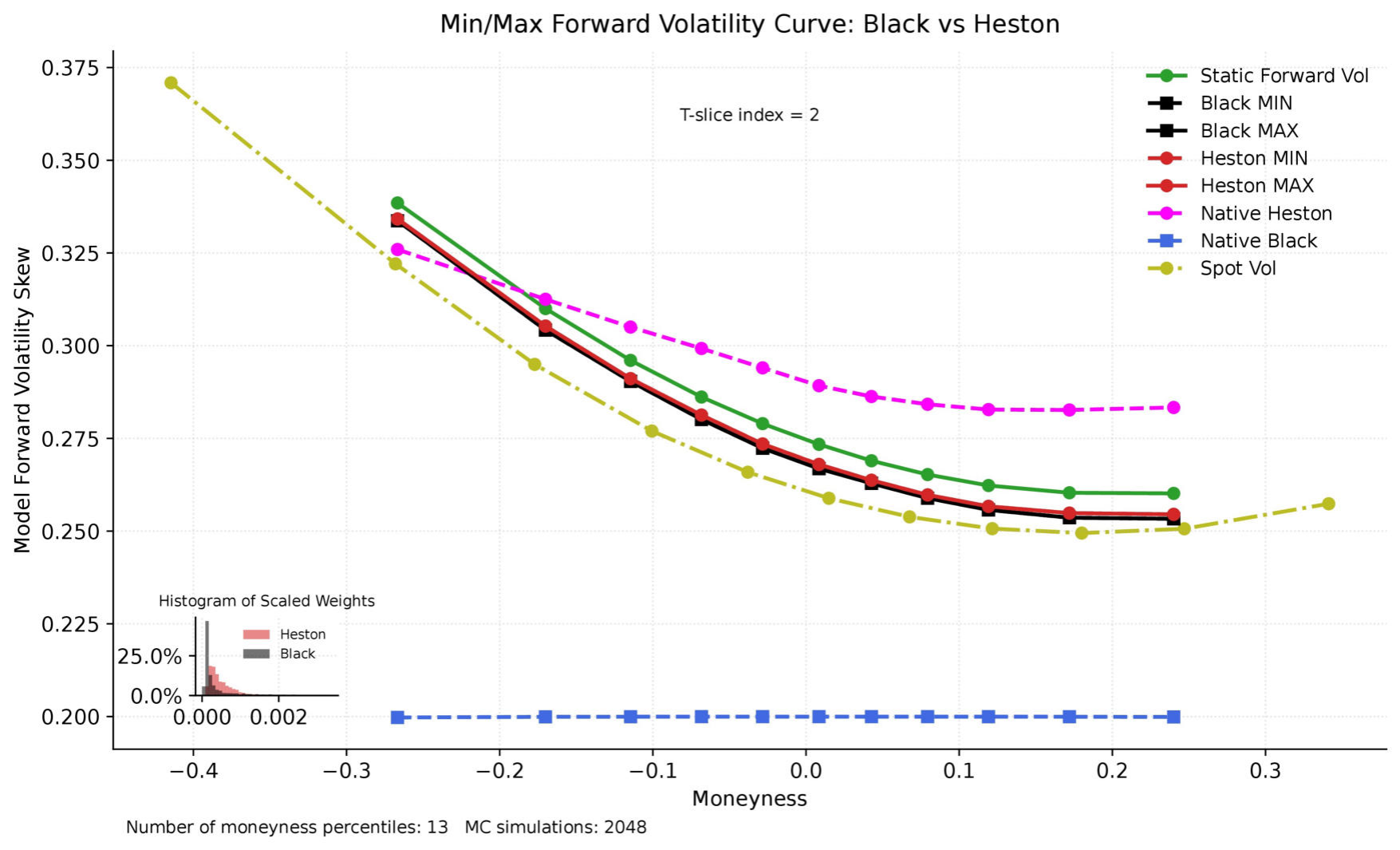}
    \caption{Heston vs Black}
    \label{fig_forward_surfaces_implied_heston_vs_black}
  \end{subfigure}
  \hfill
  \caption{Forward-volatility skew reconstruction under 
  different generators. Panel (a): Heston vs Merton; 
  panel (b): Heston vs Black. For each prior, the min-max reconstruction 
  (coloured curves) collapses onto a narrow band, illustrating 
  generator-invariant forward skew once the constraints and variance 
  penalties are imposed.
  }
  \label{fig_forward_surfaces_implied}
\end{figure}
\FloatBarrier

For a given generator $g$ and a fixed number of fixing dates $l$, denote by
\[
  D_{\min}^{(g)}(l), 
  \qquad 
  D_{\max}^{(g)}(l)
\]
the lower and upper prices delivered by the optimisation. The associated 
model mid is
\[
  m^{(g)}(l) := \frac{1}{2}\Bigl(D_{\min}^{(g)}(l) + D_{\max}^{(g)}(l)\Bigr).
\]

\paragraph{Relative intra model spread.}
The absolute intra model band is 
\[
  S_{\mathrm{intra}}^{(g)}(l) := D_{\max}^{(g)}(l) - D_{\min}^{(g)}(l).
\]
We report instead the \emph{relative} intra model spread, normalised by 
the mid and expressed in percent:
\[
  S_{\mathrm{intra,rel}}^{(g)}(l)
  :=
  \frac{D_{\max}^{(g)}(l) - D_{\min}^{(g)}(l)}
       {m^{(g)}(l)}
  \times 100.
\]
From the viewpoint of a trader committed to model $g$, 
$S_{\mathrm{intra,rel}}^{(g)}(l)$ is the natural bid-ask width on the 
Reverse Cliquet, measured as a percentage of the model mid.
\FloatBarrier
\begin{figure}[t]
  \centering
  \begin{subfigure}[t]{0.32\linewidth}
    \centering
    \includegraphics[width=\linewidth]{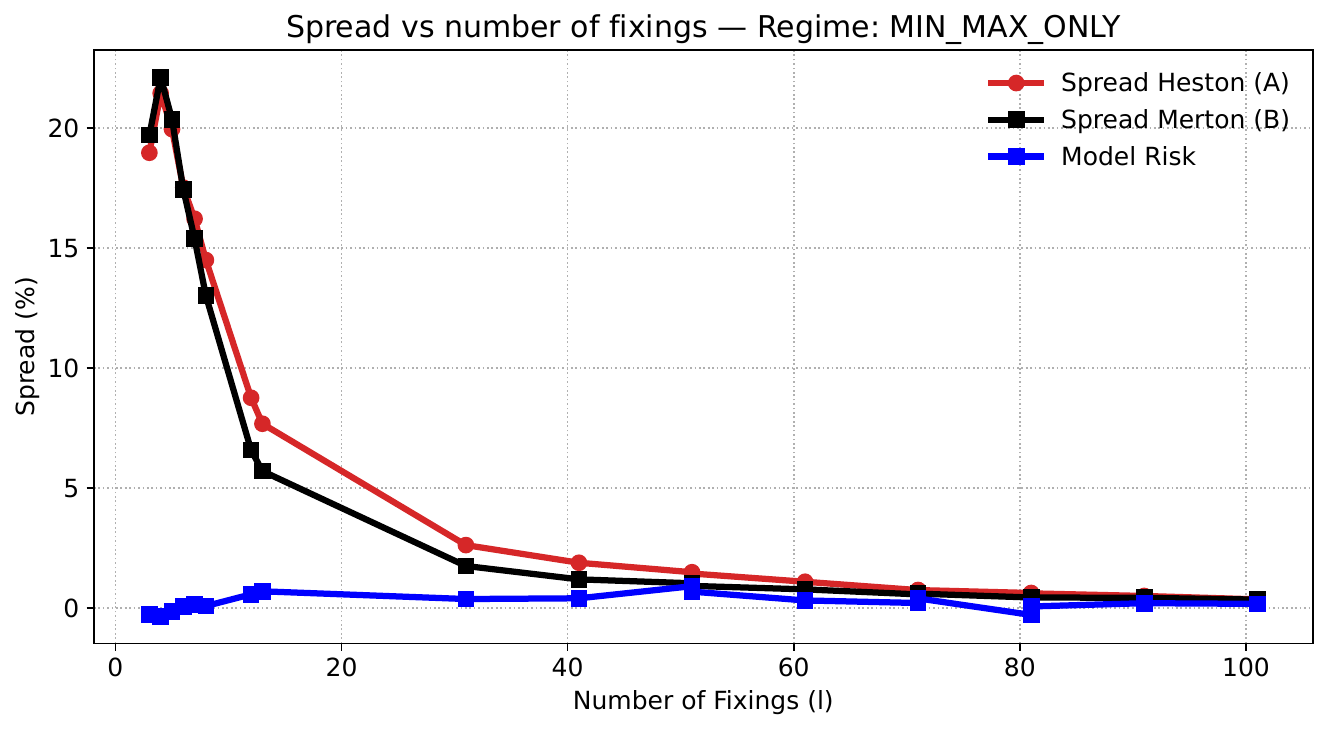}
    \caption{Raw min max}
    \label{fig_rev_cliquet_minmax}
  \end{subfigure}
  \hfill
  \begin{subfigure}[t]{0.32\linewidth}
    \centering
    \includegraphics[width=\linewidth]{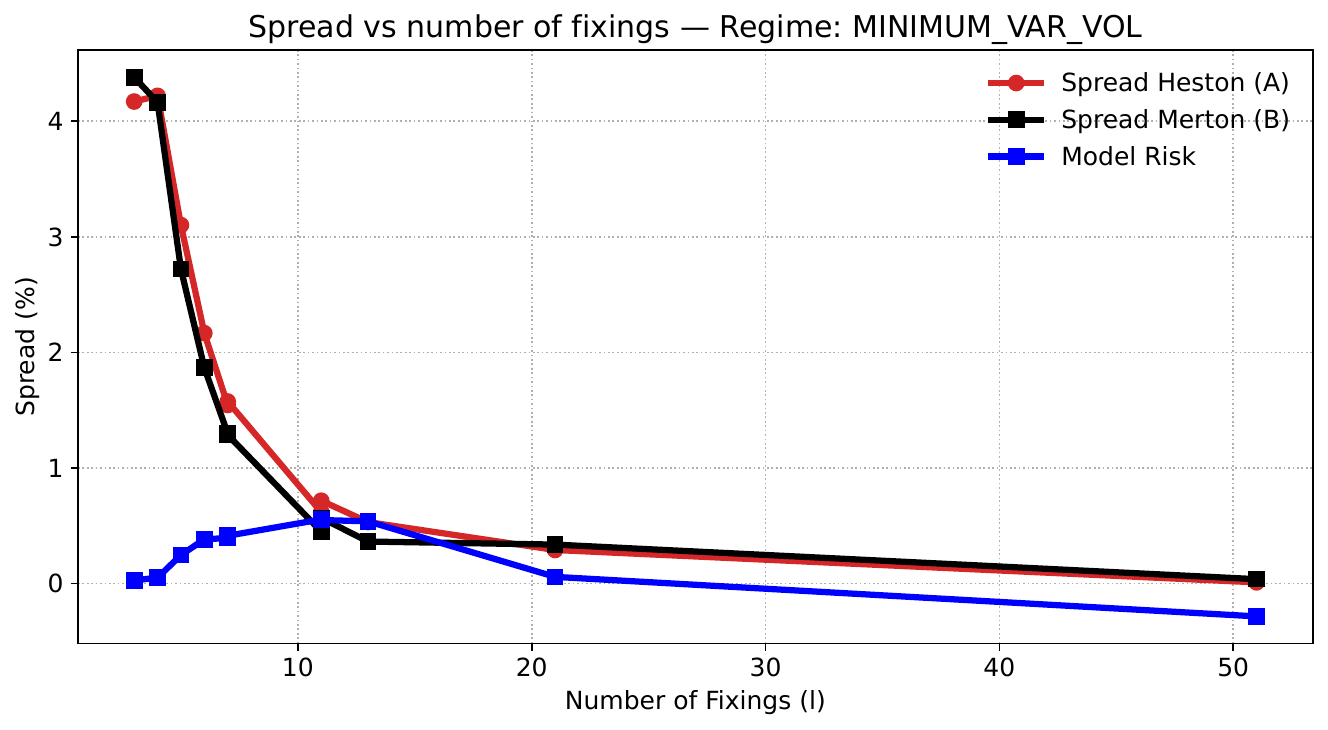}
  
    \caption{minimum dispersion of forward-vol}
    \label{fig_rev_cliquet_minvar_sigma}
  \end{subfigure}
  \hfill
  \begin{subfigure}[t]{0.32\linewidth}
    \centering
    \includegraphics[width=\linewidth]{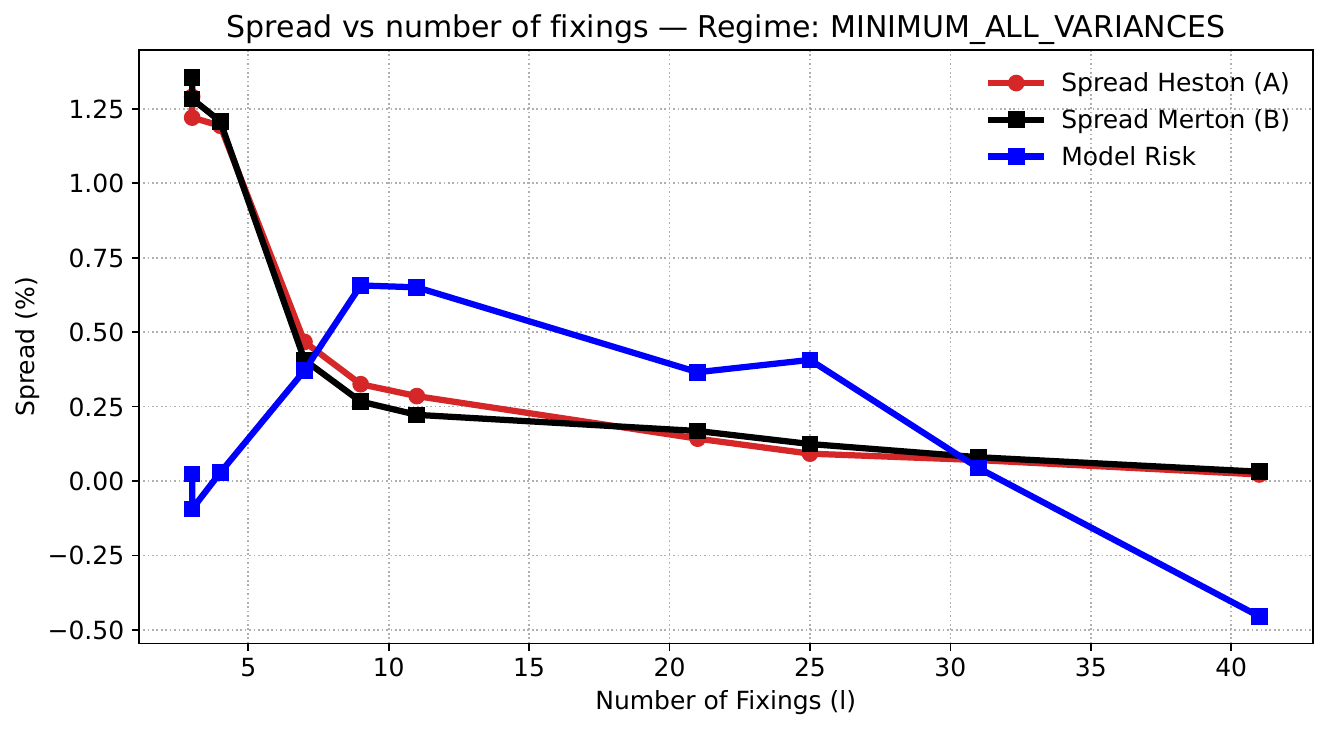}
    \caption{joint minimum dispersion of weights and forward-vol}
    \label{fig_rev_cliquet_minvar_both}
  \end{subfigure}
  \caption{The three panels report minimum and maximum reverse cliquet 
   prices for parameter set 2 (intermediate volatility regime; see Table~\ref{general_settings} 
   in Appendix~\ref{subsection_Simulation_Settings}.
  (a)  Raw min max
  (b)  minimum dispersion of forward-vol
  (c)  joint minimum dispersion of weights and forward-vol
  }
  \label{fig_reverse_cliquet_convergence}
\end{figure}
\FloatBarrier

\paragraph{Relative model risk component.}
To quantify model risk on price, we compare the mids produced by two 
generators, say Heston and Merton. The absolute model risk component is
\[
  \Delta_{\mathrm{MR}}(l)
  :=
  m^{(\mathrm{Heston})}(l) - m^{(\mathrm{Merton})}(l).
\]
We again work with a relative, symmetric normalisation:
\[
  \Delta_{\mathrm{MR,rel}}(l)
  :=
  \frac{
    m^{(\mathrm{Heston})}(l) - m^{(\mathrm{Merton})}(l)
  }{
    \tfrac{1}{2}\Bigl(m^{(\mathrm{Heston})}(l) + m^{(\mathrm{Merton})}(l)\Bigr)
  }
  \times 100.
\]
The quantity $\vert\Delta_{\mathrm{MR,rel}}(l)\vert$ measures the relative 
distance, in percent, between the mids of two equally plausible models, 
normalised by their average level. The same construction can be applied to 
any pair of priors, including the constant volatility plug in model.

\subsubsection*{Three optimisation regimes and uncertainty decomposition}

\begin{itemize}

  \item \textbf{Raw min max.}
        Only vanilla constraints are imposed. The admissible price interval 
        $[D_{\min}^{(g)}(l),D_{\max}^{(g)}(l)]$ decreases as the number of 
        fixings increases, but the relative intra model width 
        $S_{\mathrm{intra,rel}}^{(g)}(l)$ stabilises at a small but strictly 
        positive plateau. This reflects the structural fact that forward 
        skew is not determined by vanilla data alone. Empirically, for fixed 
        $l$ the quantities $S_{\mathrm{intra,rel}}^{(\mathrm{Heston})}(l)$ 
        and $S_{\mathrm{intra,rel}}^{(\mathrm{Merton})}(l)$ are very close, 
        so the intra model bid-ask is essentially generator independent. 
        The relative model risk $\vert\Delta_{\mathrm{MR,rel}}(l)\vert$ 
        remains of the same order as $S_{\mathrm{intra,rel}}^{(g)}(l)$ and 
        shrinks with $l$, but does not vanish.

  \item \textbf{Forward volatility variance minimisation.}
        Penalising the dispersion of the reconstructed forward volatility 
        slice reduces the admissible interval for each generator and forces 
        their forward volatility surfaces to almost coincide. On prices, 
        this has two effects. First, the relative intra model width 
        $S_{\mathrm{intra,rel}}^{(g)}(l)$ shrinks further and remains nearly 
        identical across generators. Second, the relative model risk 
        $\vert\Delta_{\mathrm{MR,rel}}(l)\vert$ becomes very small on the 
        natural scale of a few percent of the option premium, corresponding 
        to a few cents per 100 notional. The total residual uncertainty can 
        now be read as the combination of a narrow intra model bid-ask and a 
        modest, but still visible, model risk component.

  \item \textbf{Joint minimisation of weight variance and forward volatility variance.}
        In the third regime we minimise simultaneously the variance of the 
        weights and the variance of the forward volatility slice. 
        Conceptually, if these two variance objectives were perfectly aligned, 
        the feasible set would collapse to a single, most parsimonious 
        calibrated configuration. In that idealised case the admissible band 
        would reduce to a single point,
        \[
          S_{\mathrm{intra,rel}}^{(g)}(l) = 0, \qquad
          D_{\min}^{(g)}(l) = D_{\max}^{(g)}(l) = D_{\ast}^{(g)}(l),
        \]
        and there would be no intra model bid-ask ambiguity at all.

        In practice the dispersion of $w$ and the dispersion of $\sigma$ pull 
        in slightly different directions. There is no single weight vector 
        that minimises both objectives exactly, and the optimisation must 
        choose a compromise. When this compromise is fixed for a simple 
        reference pay-off and then used to price the Reverse Cliquet, the 
        relative intra model width $S_{\mathrm{intra,rel}}^{(g)}(l)$ becomes 
        numerically negligible at the scale relevant for trading. What 
        remains visible in the plots is the relative model risk 
        $\vert\Delta_{\mathrm{MR,rel}}(l)\vert$, that is, the gap between 
        $D_{\ast}^{(\mathrm{Heston})}(l)$ and $D_{\ast}^{(\mathrm{Merton})}(l)$ 
        expressed as a percentage of their average level. The residual spread 
        is therefore entirely concentrated into model risk.

\end{itemize}

This mechanism is entirely different from the residual width observed in the 
raw min max regime, which originates from the structural inability of vanilla 
options to determine forward-volatility skew. Under joint variance minimisation, the 
residual band is instead the consequence of the internal tension between the 
two variance criteria and manifests itself as a pure model risk component.

\subsubsection*{Computational highlight: one hundred fixing dates}

To further stress the method, we consider a configuration with one hundred 
fixing dates, corresponding to roughly one thousand linear constraints. 
This setting pushes the dimensionality and numerical stiffness of the 
problem to an extreme level. Despite this, the solver delivers stable and 
smooth extremal forward volatility surfaces and a relative price interval 
for the Reverse Cliquet well below one percent of the normalised pay-off. 
Achieving such robustness at this resolution is a strong indication of the 
scalability of the method.

\subsubsection*{Numerical outcome}

Across all regimes the admissible price interval contracts as the number of 
fixings increases. In the first two regimes, both the relative intra model 
width $S_{\mathrm{intra,rel}}^{(g)}(l)$ and the relative model risk 
$\vert\Delta_{\mathrm{MR,rel}}(l)\vert$ become small but remain clearly 
non zero. In the joint variance regime the construction concentrates the 
residual uncertainty almost entirely into the model risk component, with 
$S_{\mathrm{intra,rel}}^{(g)}(l)$ essentially zero on the relevant trading 
scale.

The Reverse Cliquet thus provides a transparent illustration of how the 
Smart Monte Carlo framework produces generator invariant prices and how the 
remaining uncertainty can be decomposed, in desk language, into an intra 
model bid-ask component and a model risk component that can, in principle, 
be computed for any exotic pay-off admitting min max bounds.

\subsection{Inversion Problem}

\subsubsection{The Ideal Inversion Problem: Uniqueness from First Principles}
\label{subsec:B1-ideal}

In the ideal setting, the Monte Carlo paths are generated by a process that already matches all observed vanilla prices. In this case the calibration constraints determine a single admissible probability measure. The weights $w_i$ are required to satisfy
\[
  w_i \ge 0, 
  \qquad \sum_i w_i = 1,
  \qquad \sum_{i=1}^{N} w_i \, \text{Payoff}_i(K,T) = P_{\text{mkt}}(K,T)
\]
for every strike-maturity pair (K,T). These conditions define a convex and typically low-dimensional feasible region.

If the simulated paths reproduce the entire vanilla surface exactly, the feasible 
region collapses to the barycentric point $w_i = 1/N$. In this ideal case the 
calibration recovers the uniform distribution and the underlying stochastic dynamics are uniquely identified from market prices. The barycentric solution is therefore the signature of a model that spans the correct pay-off space and requires no further adjustment.

Whenever the simulated paths do not reproduce the vanilla surface exactly, the feasible region becomes nontrivial and the weights deviate from the barycentric configuration. This deviation reflects structural mismatch: the paths do not encode the complete information carried by the market smiles. Large departures from 1/N indicate misspecification or lack of descriptive richness, while small departures signal near-identification.

In summary, the ideal inversion problem has a unique solution precisely when the simulated paths already satisfy the vanilla constraints. Uniform weights reveal that the generator is fully compatible with the market data, while non uniformity signals that the inversion requires an additional selection principle, developed in the next subsection.

\subsubsection{Practical Inversion: Minimising Weight Dispersion}
\label{subsec:B1-practical}

Outside the ideal case of exact vanilla replication, the weights cannot be
expected to remain uniform. Calibration noise, bid-ask fluctuations, 
discretisation effects, and generator limitations all contribute to 
deviations from the barycentric configuration. The objective of practical
inversion is therefore to select, among all admissible weight vectors, the one
that stays as close as possible to the ideal point.

To measure this deviation, we minimise the variance of the weights:
\[
\min_{w}\; \operatorname{Var}_{i}\!\left(w_i - \frac{1}{N}\right)
\]
subject to all vanilla replication and no arbitrage constraints. The problem is
convex and admits a unique minimiser.

A small variance indicates that the simulated paths already capture most of the
information encoded in the vanilla surface and that the underlying process is
nearly identified. A large variance reveals structural mismatch: the path set
does not span the correct pay-off space and requires substantial reweighting to
match market data. The dispersion of the weights therefore provides a direct
and quantitative indicator of model adequacy.

\paragraph{Statistical meaning of the min var criterion.}
Let $p_i$ denote the discounted pay-off on path $i$. The weighted estimator
\[
  \widehat{P}(w) = \sum_i w_i p_i
\]
has leading order standard deviation
\[
  \sigma_{\text{price}}
  \approx
  \sqrt{
        \sigma_P^2
        + \bar{p}^2 \sigma_W^2
        + 2 \bar{p} \, \text{Cov}(P,W)
      },
\]
where $\sigma_W$ is the standard deviation of the weights. Since the contribution
proportional to $\sigma_W$ dominates whenever the pay-off distribution is non
degenerate, the estimator error increases with the dispersion of the weights.
Minimising weight variance therefore stabilises the estimator and reduces its
sensitivity to perturbations.

The same criterion that selects the weight vector closest to the barycentric
configuration also minimises the statistical error of the weighted Monte Carlo
estimator. Practical inversion and estimator robustness are therefore aligned
under the min var principle.

\subsubsection{Conditional inversion: pinning an additional observable}

The inversion principle extends naturally to cases in which an additional
linear observable must be matched. Let $C(w) = c^{*}$ be any linear functional 
of the scenario-weights. Adding this condition to the feasible set,
\begin{equation}
   \min_{w}\ \text{Var}(w_i - 1/N)
   \qquad\text{s.t.}\quad
   C(w) = c^{*},
  \label{eq_inversion_theorem_min_variance_COND}
\end{equation}
preserves convexity and yields a unique solution. The fixed-point loop used in 
Subsections \ref{subsec:B1-ideal}~-~\ref{subsec:B1-practical} applies without modification.

This conditional inversion mechanism allows one to incorporate additional
moment constraints or structurally meaningful indicators (for instance forward
variance at a given node) while retaining the uniqueness and stability of the
min-variance solution. It provides a flexible extension of the inversion 
framework whenever one wishes to anchor the calibrated distribution to one 
extra observable.

\subsection{Extending the Framework and G\"odel-type Incompleteness}
\label{subsec:B2}

Vanilla replication determines the admissible weight vector only up to a low dimensional family of solutions. Exotic valuations, however, depend on additional latent variables that are not fixed by vanilla constraints alone. To handle these degrees of freedom in a principled way, we extend the inversion framework by introducing additional state variables and corresponding variance penalties. This enrichment increases the expressive power of the optimisation while preserving convexity, and reveals an intrinsic form of structural incompleteness that persists at every level.

\subsubsection{Double Variance Objective and Enlargement of the State Space}

Vanilla constraints determine the weights w only partially. When exotic valuations 
depend on latent quantities such as the forward-start volatility slice $\sigma$, 
several $\sigma$-configurations may still fit the same vanilla surface. 
To stabilise these latent coordinates and remove the residual degrees of freedom, we introduce a second variance penalty and solve
\[
  \min_{w,\sigma}
      \text{Var}[w]
      + \lambda_{\sigma} \, \text{Var}(\sigma - \sigma^{prev})
\]
subject to all vanilla replication and no arbitrage constraints, and to the affine relations linking $\sigma$ to the forward-start prices.

The first term controls the dispersion of the weights and keeps the 
solution close to the ideal inversion of Subsection~\ref{subsec:B1-ideal}. 
The second term controls the dispersion of the forward-volatility 
slice and ensures that the update $\sigma \rightarrow \sigma^{prev}$ 
remains smooth across iterations. Both terms are convex quadratic 
forms, and the combined optimisation remains a convex SOCP.

Introducing $\sigma$ enlarges the descriptive language of the method. 
Once $\sigma$ is part of the state vector, the framework can 
represent and constrain forward-start volatilities directly, something 
that is impossible with weights alone. The optimisation then selects 
a single representative element among all $\sigma$ that are 
compatible with vanilla prices.

The same mechanism applies to additional latent coordinates. For example, non adjacent forward-start volatilities, conditional forward-start structures, or barrier-related quantities can be represented by introducing a third block $\gamma$ and a corresponding variance penalty. Each new block increases the expressive power of the framework while preserving convexity, and each penalty removes the new degrees of freedom created by the enlargement.

\subsubsection{G\"odel-type Incompleteness Analogy}

Even after introducing several variance-penalised blocks, the framework does not become fully complete. Each enlargement of the state vector (for example by adding $\sigma$, then $\gamma$, and so on) expands the set of pay-offs and volatility functionals that can be represented. However, new products can always be constructed that lie outside the enlarged language and would require yet another block of variables to be captured. No finite construction stabilises the process.

This behaviour parallels the classical incompleteness phenomenon introduced by Kurt G\"odel. In G\"odel's setting, a formal system based on a finite set of axioms can generate many theorems, but can never capture all truths expressible in its own language. Extending the system by adding new axioms enlarges the set of provable theorems but does not eliminate incompleteness; new undecidable statements inevitably appear.
This analogy is purely structural and not meant as a formal correspondence; it serves only to emphasise the persistent incompleteness of any finitely-extended calibration framework.

In our context, the calibration constraints and variance penalties play the role of axioms, while the admissible calibrated configurations ($w$, $\sigma$, $\gamma$, ...) correspond to the theorems derivable from these axioms. Adding a new block of latent variables is equivalent to enriching the formal language: it resolves some degrees of freedom but creates a larger space in which incompleteness reappears. No finite number of extensions can encode every conceivable pay-off or implied-volatility functional.

This hierarchy of expressive levels is summarised in Table~\ref{table:model_richness}. Level I corresponds to weights alone; Level II includes the forward-volatility slice $\sigma$; Level III includes an additional block $\gamma$. Higher levels follow the same pattern: each enlargement increases expressiveness but never removes incompleteness entirely.
\begin{table}[htbp]
  \centering
  \setlength{\tabcolsep}{6pt}
  \renewcommand{\arraystretch}{1.15}

  \caption{Model richness versus solver complexity}
  \label{table:model_richness}

  \begin{tabular}{@{} l l l l @{}}
    \toprule
    \textbf{Level} & \textbf{Decision variables} &
    \textbf{Products captured} & \textbf{Solver class / cost} \\
    \midrule
    I   & $w$               & Linear pay-offs: vanilla prices, forward prices & LP (low) \\[2pt]
    II  & $w,\sigma$        & Adjacent forward-start implied vols            & SOCP (medium) \\[2pt]
    III & $w,\sigma,\gamma$ & Conditional/non-adjacent forward vols,\newline barrier vols & enlarged SOCP (high) \\
    \bottomrule
  \end{tabular}
\end{table}

\subsection{Isotropic Monte Carlo}
\label{subsec:isotropic_MC}

In a standard Monte Carlo simulation with uniform weights, the empirical
distribution of paths is statistically symmetric and exhibits no preferred
direction. When weights are optimisation variables, this symmetry is no longer
guaranteed. Since extremal valuations are highly sensitive to the behaviour of
the weights, it is natural to ask whether one can impose a meaningful notion of
Monte Carlo isotropy.

\subsubsection{Concept and limitations of isotropy}

Given any physically relevant ordering of the paths (for example by terminal
underlying value), define the cumulative mass profile
\[
f(n) = \sum_{i=1}^{n} w_{(i)} .
\]
For a uniform Monte Carlo sample, one expects
\[
f_{0}(n) = n/N
\]
up to sampling fluctuations of order $\sqrt{x(1-x)/N}$ with $x=n/N$. This
motivates the idea of keeping the empirical distribution inside a fluctuation
envelope that mimics a standard Monte Carlo simulation.

However, enforcing such an envelope along a single ordering direction is
insufficient. True isotropy is inherently multi-dimensional: to be meaningful,
the envelope would need to hold simultaneously along every projection or
rotation of the path cloud. This would require an exponentially large family of
constraints and is therefore computationally infeasible. In practice, when one
imposes isotropy only along one direction, the optimiser satisfies the envelope
by concentrating the mass at one extreme of the sorted sample, producing a
degenerate and non-physical solution. This behaviour is structural: a free weight
optimisation cannot reproduce the multidirectional fluctuation structure of a
genuinely isotropic Monte Carlo ensemble.

\subsubsection{Model set equivalence}

The above limitation motivates a different diagnostic: instead of forcing
isotropy, we test whether two path sets carry the same descriptive content for
pricing and calibration. Given two sets $A$ and $B$, we build a mixed pool and
enforce a mass splitting constraint
\[
\sum_{i\in A} w_i = \lambda , \qquad 0 \le \lambda \le 1 .
\]
If both sets span the same pay-off space under the imposed constraints, the
optimiser returns $\lambda \approx 1/2$. If $\lambda$ moves towards $0$ or $1$,
one set dominates the other. Under linear objectives one expects
\[
f^{star} \approx \lambda f^{star}_{A} + (1-\lambda) f^{star}_{B} ,
\]
so that $\lambda$ directly measures the marginal contribution of each subset.
Impoverishing one set (for example by repeating paths) deliberately moves
$\lambda$ away from $1/2$, providing a robustness check. This criterion serves as
a practical replacement for multi directional isotropy, capturing equivalence at
the level that matters operationally.

\subsubsection{A workable alternative: the interior barrier}

Since exact isotropy is not achievable in a practical or meaningful way, we use
a softer mechanism that preserves statistical balance without enforcing a
multi-dimensional fluctuation envelope. An interior logarithmic barrier keeps
all weights strictly within the interior of the simplex and suppresses 
non-physical spikes or boundary collapse. The barrier is purely geometric: it has no
probabilistic interpretation, does not refer to any prior distribution and does
not induce any entropic interpolation. Its implementation and derivatives are
reported in Appendix~\ref{app_internal_barrier}.

\medskip

In summary, exact isotropy is conceptually well defined but operationally
intractable. The combination of the mass splitting test and the interior barrier
provides the practically relevant counterpart: it prevents degeneracy, ensures
balanced weights and supports the stability of the extremal valuations developed
in Section~\ref{Section_B_inversion}.


\section{Computational Performance}
\label{section_C_performance}

The analysis of numerical performance focuses on two complementary aspects: 
(i) how the precision of the method improves with the number of Monte Carlo scenarios, 
and (ii) how the computational cost scales with the main problem dimensions, namely 
the number of scenarios, the number of vanilla constraints, and the number of fixing dates. 
Together, these scaling laws define the efficiency envelope of the Smart-MC algorithm.

\subsection{Error scaling and precision}

The central question is how numerical precision evolves as the number $N$ of simulated scenarios increases. For a fixed estimator, classical Monte Carlo theory implies that the root-mean-square error (RMSE) decays as
$ \mathrm{RMSE}(N) \approx C/{\sqrt{N}}$
where $C$ is a problem-dependent constant determined by the pay-off variance (Central Limit Theorem; see Glasserman~\cite{Glasserman2003}). In our constrained framework, pricing conditions (such as matching vanilla and barrier pay-offs) act as control variates and reduce this constant for any given $N$, but they do not alter the fundamental $N^{-1/2}$ behaviour.

The nontrivial effect comes from the optimization step used to compute the scenario weights. As $N$ increases, the dimension of the optimization problem grows and the solver requires more work per additional scenario. Consequently, the effective constant in the Monte Carlo error becomes a slowly varying function $C(N)$, and the RMSE can be written as
\[
\mathrm{RMSE}(N) \approx \frac{C(N)}{\sqrt{N}}.
\]
Over the range of $N$ relevant for calibration, this dependence is well captured by a power law $C(N) \propto N^{\alpha}$ with $\alpha \ge 0$, which leads to the effective scaling
$ \mathrm{RMSE}(N) \sim N^{-\left(\tfrac{1}{2} - \alpha\right)}$,
to be interpreted as a correction to the benchmark $N^{-1/2}$ rate. The parameter 
$\alpha$ summarizes the additional computational burden induced by the solver

Moreover, if the Monte Carlo experiment is replicated $M$ times with 
independent random seeds, each producing $N$ paths, the corresponding 
weight vectors can be concatenated into a single feasible sample of 
size $MN$. The variance of the resulting estimator then scales 
as $1/\sqrt{MN}$, while the total wall-clock time grows approximately 
linearly in $M$, since the $M$ optimizations can be executed in parallel.
In practice, a small number of such replications is usually sufficient to achieve
stable and accurate results (see Glasserman~\cite{Glasserman2003}).

\begin{figure}[t]
  \centering
  \begin{subfigure}[t]{0.27\linewidth}
    \centering
    \includegraphics[width=\linewidth]{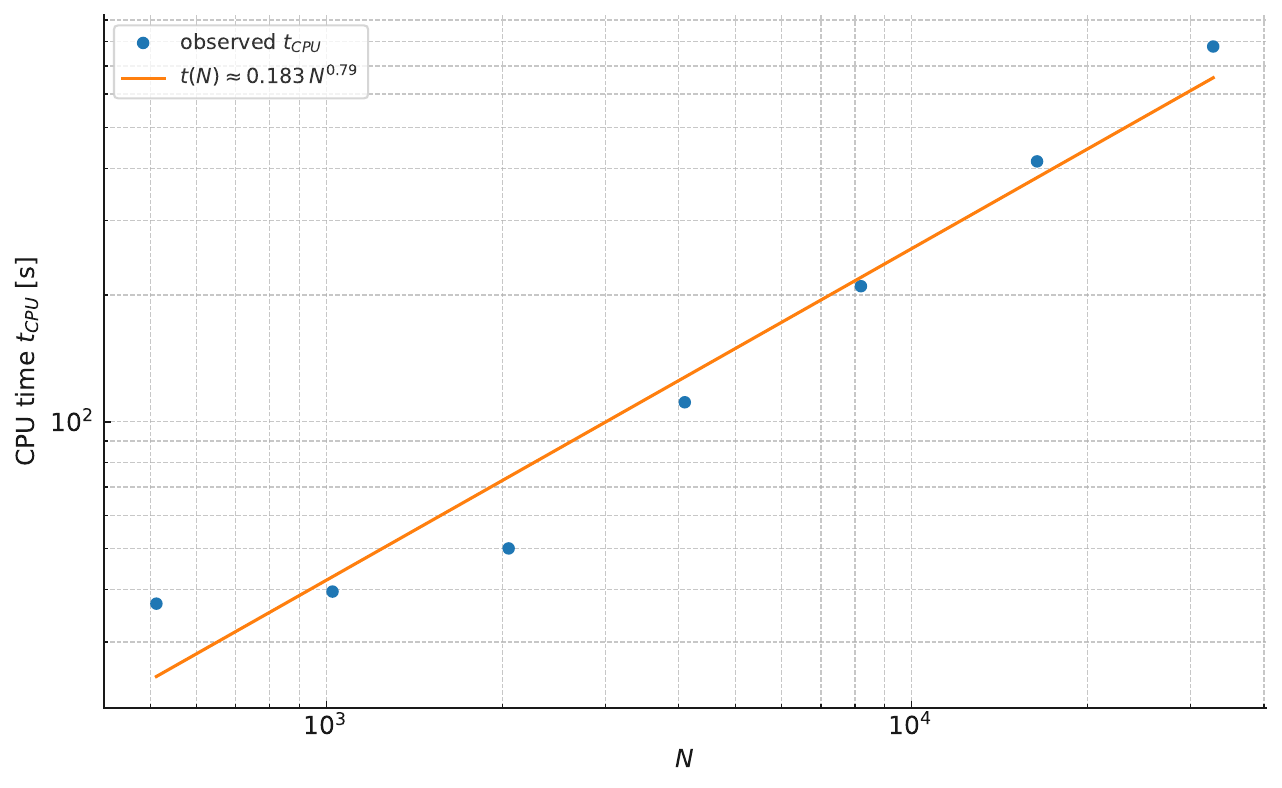}
    \caption{Runtime vs.\ scenarios $N$ (log--log scale).\\
             Power--law fit $\text{Time}(N) \approx 0.183\, N^{0.79}$.\\
             Other parameters: strike grid size $m=11$, fixing dates $l=10$.}
    \label{fig_runtime_N}
  \end{subfigure}
  \hfill
  \begin{subfigure}[t]{0.32\linewidth}
    \centering
    \includegraphics[width=\linewidth]{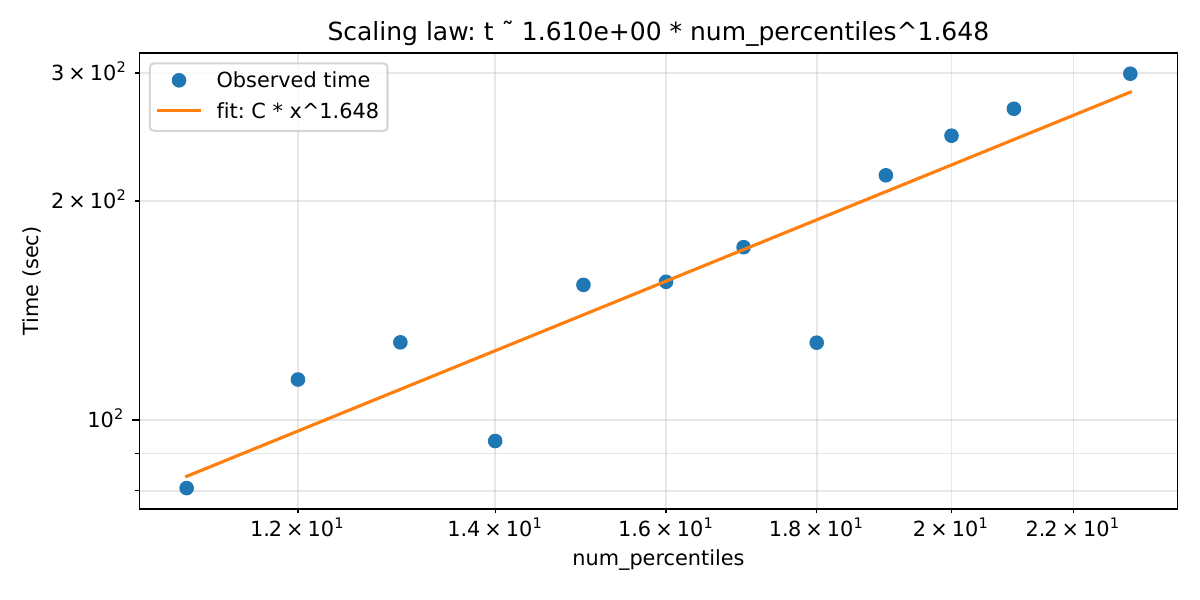}
    \caption{Runtime vs.\ strike grid size $m$ (log--log scale).\\
             Power--law fit of the form $\text{Time}(m) \approx 1.61\, m^{1.648}$.\\
             Other parameters: $N = 2^{12}$, fixing dates $l=10$.}
    \label{fig_runtime_m}
  \end{subfigure}
  \hfill
  \begin{subfigure}[t]{0.32\linewidth}
    \centering
    \includegraphics[width=\linewidth]{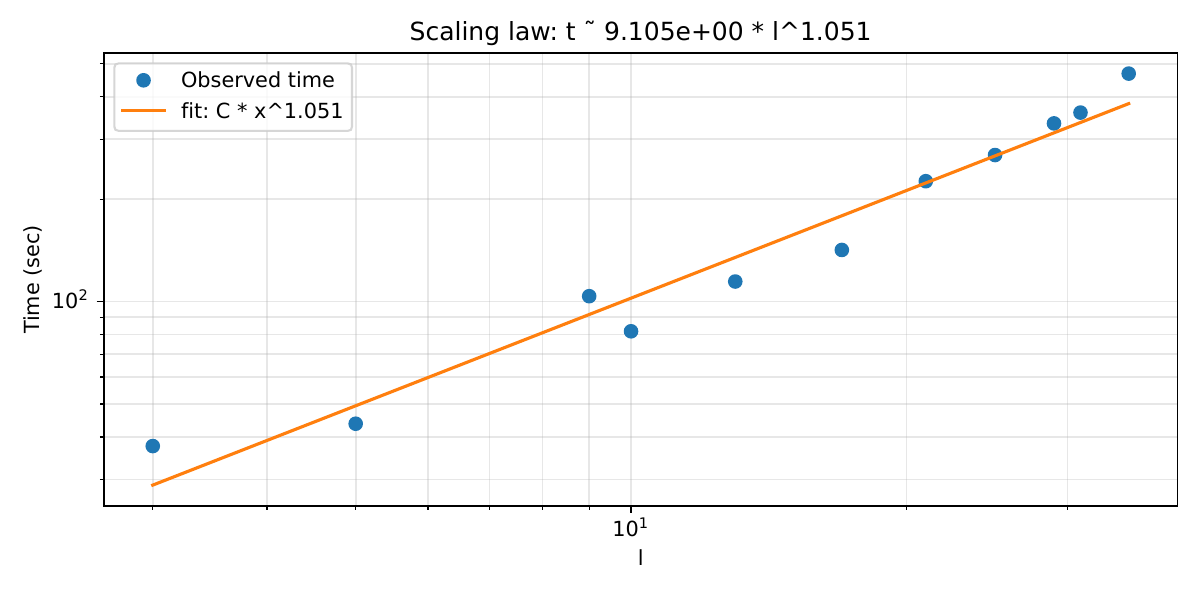}
    \caption{Runtime vs.\ number of fixing dates $l$ (log--log scale).\\
             Power--law fit $\text{Time}(l) \approx 9.105\, l^{1.051}$ with an exponent close to one.\\
             Other parameters: $N = 2^{12}$, strike grid size $m=11$.}
    \label{fig_runtime_dates}
  \end{subfigure}
  \caption{Wall-clock runtime (seconds) for the sequential Monte Carlo
           calibration as a function of:
           \textbf{(a)} Monte Carlo scenarios $N$;
           \textbf{(b)} number of vanilla calibration points $m$;
           \textbf{(c)} number of fixing dates $l$.
           Each panel shows observed runtimes and the corresponding
           log--log power-law fit, with coefficients reported in the captions.}
  \label{fig_runtime_macro}
\end{figure}

\subsection{Time scaling}

This subsection analyzes the time scalability of the Smart--MC algorithm along three orthogonal axes:
\begin{itemize}
  \item the number of Monte Carlo scenarios $N_{MC}$,
  \item the strike grid size $m$,
  \item the number of fixing dates $l$.
\end{itemize}
All timings were obtained on an Intel Ultra 185H processor using 
the open-source solver Clarabel, with identical code and hardware settings across all tests.

\subsubsection{Time scaling with Monte Carlo scenarios}

Empirically, the wall-clock time grows according to a power law
\[
  \text{Time}(N_{MC}) \approx c_N\, N_{MC}^{\,p_N},
\]
estimated by least-squares regression in log--log coordinates.  
Over the range $N_{MC} = 2^M$, $M = 9,\dots,15$, the fitted exponent is
$p_N \approx 0.79 < 1$ (Figure~\ref{fig_runtime_N}). This is consistent with a near-linear cost for path
generation combined with a modest overhead from the sparse interior-point phase, and confirms that the
algorithm remains affordable even for large Monte Carlo ensembles.

\subsubsection{Time scaling with the strike grid size}

For fixed Monte Carlo scenarios and number of fixing dates, solver time grows essentially linearly with the
number of vanilla calibration points $m$ until matrix factorization costs dominate. The timings in
Figure~\ref{fig_runtime_m} are well described by a power law
\[
  \text{Time}(m) \approx c_m\, m^{p_m},
\]
with an exponent $p_m$ close to one in the regime of interest. In practice, a grid with $m \le 32$ strikes is
sufficient to capture typical market smiles, and the fitted power law confirms that the dependence on smile
granularity remains mild over this range.

\subsubsection{Time scaling with the number of fixing dates}

Similarly, empirical timing tests for the number of fixing dates are 
well approximated by
\[
  \text{Time}(l) \approx c_l\, l^{p_l},
\]
with an almost linear exponent $p_l$ close to one (Figure~\ref{fig_runtime_dates}). 
This slow growth of runtime with $l$ is crucial for long-dated 
reverse cliquet structures, as it allows us to handle a large 
number of fixing dates without incurring prohibitive computational 
costs.

Pushing the framework further, we solve a \textbf{100-date} instance whose
instrument set contains
\[
100 \text{fixing dates} \times 20 PV  = 2000 \text{ pay-off constraints}
\]
(ten vanilla strikes plus ten digital co-strikes per fixing date, the
latter injecting slope information that sharpens the implied density
curvature). The full problem converges in under \textbf{2 h} wall-clock on a
standard end-user CPU (hardware profile in
Appendix~\ref{subsection_Simulation_Settings}) using only open-source
components, with no GPU acceleration.

Handling more than 100 fixing dates highlights the computational strength and
scalability of the proposed algorithm.

The min and max arbitrage-free prices converge as the number of fixing dates 
increases, but the gap does not vanish. This residual spread reflects genuine
model incompleteness: forward skew uncertainty is not determined by vanilla
options alone, and the reverse cliquet pay-off aggregates a non-linear functional
of pathwise forward increments, which remains non-replicable even in the
continuous fixing limit~\cite{Hobson1998,Cox2011,BadikovEtAl2017}.

In summary, the observed near-linear growth in runtime with $l$, the ability
to ingest 1000 pay-off constraints, and the joint treatment of the volatility
surface and path-dependent pay-offs establish the method as a practical tool
for real-time XVA and risk-of-risk analytics on high-dimensional exotics.

\subsection{Global scaling laws}

The overall behaviour of the Smart--MC algorithm can be summarised by two
scaling relations, equations~\eqref{eq_scaling_cpu_time}
and~\eqref{eq_scaling_error_for_RMSE}.

\paragraph{(1) Runtime scaling.}

\begin{equation}
T(N, m, l) = c\, N^{\,1+\alpha_N} \, m^{\,1+\alpha_m} \, l^{\,1+\alpha_l},
\label{eq_scaling_cpu_time}
\end{equation}
where the exponents $\alpha_N$, $\alpha_m$ and $\alpha_l$ quantify deviations
from ideal linear scaling in each direction. Empirically, these corrections are
small in magnitude, and the fits extracted from
Figure~\ref{fig_runtime_macro} are consistent with a runtime that grows with
a power-law exponent close to one in each direction, i.e.\ nearly linearly in
$N$, $m$ and $l$ over the tested range.

\paragraph{(2) Error scaling.}

\begin{equation}
\mathrm{RMSE}(N) \sim N^{-\left(\tfrac{1}{2} - \alpha\right)},
\label{eq_scaling_error_for_RMSE}
\end{equation}

\section{Future Developments}

The calibration structure developed in Section~\ref{Section_B_inversion} opens several directions
for future research. We outline the most promising ones below.

\subsection{From extremal prices to full histograms.}
The min-max envelope computed by the Smart-MC framework corresponds to the
extreme points of the feasible polytope. Because the feasible set is convex,
any convex combination of admissible solutions is also admissible. This makes
it possible to replace the extremal interval by a full probability distribution
for any quantity of interest. For instance, by sampling multiple admissible
solutions and aggregating them, one can obtain an empirical histogram for the
price of an exotic or for any forward-variance quantity. Such distributions
provide uncertainty quantification beyond worst-case bounds and deliver explicit
statistical error bars.

\subsection{Dynamic Forward Variance Reconstruction}

Once the forward variance slice is uniquely determined by the double variance 
calibration, it becomes a genuine state variable of the system. This raises a 
natural question: can one evolve the pair $(S_{t}, V_{\mathrm{fwd}})$ jointly 
inside a Monte Carlo engine, obtaining scenario paths whose forward smile 
dynamics remain fully consistent with market implied constraints?

The idea is to update, at each time step, both the asset level $S_{t}$ and the 
calibrated forward variance slice, resolving, at each step, a 
lightweight weighted Monte Carlo update 
that preserves vanilla consistency along the simulated trajectory. This 
would produce forward smile dynamics driven directly by the market anchored 
forward variance structure, rather than by a parametric stochastic volatility 
law.

Developing such a dynamic scheme requires further work. In particular:
\begin{itemize}
\item the stability of repeated reweighting under time evolution must be 
      analysed;
\item the interaction between forward variance increments and pathwise drift 
      adjustments must be formalised;
\item the resulting process should be compared with classical stochastic volatility 
      and local volatility dynamics.
\end{itemize}

This line of research would extend the present static framework into a fully 
dynamic, smile consistent Monte Carlo engine, potentially enabling scenario 
generation, XVA computation, and risk of risk analysis under market implied 
forward smile dynamics.

\subsection{Multi-asset generalisation.}
Although this paper focuses on a single underlying, the Smart-MC framework
extends naturally to the multi-asset setting. The main challenge lies in the
construction of joint no-arbitrage constraints and consistent forward-start
structures for baskets or multi-factor underlyings. This direction represents
a natural continuation of the present work.

\section{Conclusions }

As Emanuel Derman once cautioned, "A model is at best a cartoon of 
reality." Our approach ultimately clarifies a bank's choice of one 
specific "cartoon," illustrating how the \textbf{dichotomy} between a broad 
set of plausible worlds and contending with an infinite 
model-independent perspective can be \textbf{overcome} by a unique 
model-independent valuation, which at the same time bridges robust 
pricing envelopes with the inversion problem and closes the loop 
between theory and front-office implementation.

\newpage
\clearpage
\appendix

\renewcommand{\thesection}{App-\Alph{section}}

\section*{Appendices}
\addcontentsline{toc}{section}{Appendices}

\subsection*{Scope, Structure \& Quick Map}
\label{app:intro}

\noindent
\textbf{Supplementary Material for the Sections.}  
Each appendix expands, proves, or benchmarks the results of its
corresponding main-text subsection:

\begin{itemize}[leftmargin=2em]
  \item \textbf{Section~\ref{Section_methodology} - Methodology} $\longrightarrow$ \textbf{Appendix App-A: Methodology Details}
  \item \textbf{Section~\ref{Section_B_inversion} - Inversion Problem} $\longrightarrow$ \textbf{Appendix App-B: Inversion Proofs}
  \item \textbf{Section~\ref{section_C_performance} - Computational Performance} $\longrightarrow$ \textbf{Appendix App-C: Performance Benchmarks}
\end{itemize}

\section{Appendix: Methodology Details} \label{app:A}

\subsection{Glossary of key terms \& Ancillary Data} 
\label{app_glossary}

  \begin{itemize}

       \item \textbf{Static (Model-Independent) Forward Volatility Surface $\sigma_{Fwd-static}(T_1,T_2)$ }  
       For an arbitrage-free spot volatility term-structure \( \sigma(T) \), the (Black) forward
       volatility between \(T_1<T_2\) follows immediately from variance additivity in the Black framework:  
       \begin{equation*}
       \sigma_{Fwd-static}(T_1,T_2)
       =\sqrt{\frac{\sigma_{spot}^{2}(T_2)\,T_2-\sigma_{spot}^{2}(T_1)\,T_1}{T_2-T_1}}.
       \label{eq_fwd_static_definition}
       \end{equation*}
       Because it depends only on the two marginal vanilla smiles, \( \sigma_{Fwd-static} \) is
       \emph{uniquely} determined once the spot surface is fixed~\cite{GlassermanWu2011}.

       \item \textbf{Model-Implied Forward Volatility Surface $\sigma_{Fwd-model}(T_1,T_2)$ }\\
       The forward-implied volatility is obtained in two steps:
       \begin{itemize}
          \item First, we compute the forward-start option price under a given
                stochastic (or LSV, rough-vol, etc.) model.
          \item Second, we find the Black volatility
                $\widetilde{\sigma}_{\text{Fwd}}(T_1,T_2)$ that, when inserted into the
                Black formula for a forward-start option, reproduces the model price.
       \end{itemize}
       Because this implied volatility is anchored to a specific dynamic model, it 
       differs from the static forward variance, except in the 
       deterministic--volatility limit.

       \item \textbf{Model-Free Implied Variance (MFIV)} \\[2pt]
       The MFIV is the total risk-neutral variance inferred directly from out-of-the-money 
       plain-vanilla option prices via the Carr--Madan log-contract identity:
       \begin{equation}
        \sigma_{\mathrm{MFIV}}^{2}(T)
        = \frac{2e^{rT}}{T}
        \int_{0}^{\infty}
        \frac{P(K,T) - C(K,T)}{K^{2}}\,dK,
       \end{equation}
       where $C(K,T)$ and $P(K,T)$ denote, respectively, call and put prices with maturity $T$ and strike $K$. 
       For the operational role of MFIV within forward-variance anchoring and the Smart-MC optimisation, 
       see Appendix~\ref{implementing_total_variance_conservation}.

       \item \textbf{Linear Programming}: an optimisation technique used to 
       solve linear problems with constraints. 
            \begin{itemize}[leftmargin=2.2em]
              \item \textbf{Interior-Point Method} - polynomial-time algorithm class
                    that follows a path of strictly feasible iterates inside the
                    constraint set toward the Karush-Kuhn-Tucker point, underpinning
                    modern large-scale LP and SOCP solvers.
            
              \item \textbf{Second-Order Cone Program (SOCP)} - optimise a linear
                    cost under affine-linear and Lorentz-cone constraints
                    $(\|A_i x + a_i\|_2 \le b_i^{\top}x + \beta_i)$.

              \item \textbf{Singular-Value Decomposition (SVD)} - factorise any
                    matrix $(M\) as \(M = U \Sigma V^{\top})$, where
                    $(\Sigma)$ collects non-negative singular values; the work-horse
                    for numerically stable rank reduction and least-squares
                    solutions.
            
            \end{itemize}

       \item \textbf{Multiverse}: \label{gloss:multiverse}
       the ensemble of feasible probability distributions 
       produced by different ways of reweighting paths. Each reweighting 
       corresponds to a valid "model" fitting the same vanilla data.

       \item \textbf{Inversion Problem}: the question of whether a single underlying 
       model can be uniquely identified from observed vanilla prices. 

       \item \textbf{Conditional Inversion}: \label{gloss:conditional_inversion}
       a specialized form of inversion that 
       includes additional constraints (e.g., a target forward-start 
       volatility). The framework then selects the unique distribution 
       within the larger multiverse that meets this extra condition.

       \item \textbf{Definition of reverse cliquet option.} \label{gloss:reverse_cliquet} 
       The option pay-off is given by:
       \begin{equation*}
       \text{Payoff} = \max\left(0,\, H + \frac{1}{n} \cdot  \sum_{i=1}^{n} \, \min\left(\frac{S_{t_i} - S_{t_{i-1}}}{S_{t_{i-1}}},\, 0\right)\right),
       \end{equation*}
       and exhibits a strong dependence on the forward volatility skew, as it aggregates 
       negative asset performances across all fixing dates and deducts them from a fixed 
       initial cap. This cumulative exposure to forward volatility skews across multiple maturities 
       introduces substantial sensitivity to smile dynamics.

  \end{itemize}

\subsection{Formulation and Implementation for non-linear constraints}
\label{app:formulation_and_implementation}

This appendix collects the technical ingredients needed to handle the
non-linear constraints that arise once forward-start volatilities and
model-free variance identities are incorporated in the Smart-MC calibration.
The guiding idea is to encode all non-linear structure into a small number
of auxiliary variables and quadratic penalties, so that each inner step
remains a convex LP or SOCP, while global non-linear consistency is enforced
by an outer fixed point loop.

We group the construction into four blocks:
(i) core constraints defining the feasible set,
(ii) a linear affine relation between reweighted prices and forward
volatilities and variances,
(iii) conservation of total variance via the Carr Madan log contract
identity, and
(iv) the SOCP implementation of the quadratic objectives associated with
weight and volatility dispersion.

\subsubsection{Core constraints}

At a fixed outer iteration, the inner calibration problem is convex and
imposes the following constraints:

\begin{itemize}[leftmargin=1.4cm]

  \item[(C1)] \textbf{Weight variance control.}
  A convex bound on $\|\boldsymbol w - \boldsymbol 1/M\|_2$, active when
  weight dispersion penalisation is enabled. This limits the distance of
  the reweighted measure from the uniform baseline and is implemented as
  a quadratic penalty and/or second order cone constraint in the SOCP
  formulation.

  \item[(C2)] \textbf{Forward volatility smoothness.}
  A convex bound on $(\sigma_{t,j} - \sigma^{\text{prev}}_{t,j})^2$
  aggregated across slices, active when forward volatility dispersion
  is penalised. This stabilises the forward volatility surface across
  iterations and is again implemented via quadratic penalties in the
  SOCP.

  \item[(C3)] \textbf{No arbitrage.}
  Monotonicity in strike and convex butterfly spread inequalities are
  enforced as linear constraints on reweighted vanilla prices, ensuring
  static no arbitrage of the calibrated surface.

  \item[(C4)] \textbf{Variance anchoring.} \label{C4_Variance_anchoring}
  Total variance at each maturity is anchored via the Carr Madan
  variance swap identity (derivation below). This enforces that the
  weighted Monte Carlo measure reproduces the market MFIV level at each
  maturity. At each intermediate time $t$ we additionally impose:
  \[
    \sum_{j} \Delta K_j \,\bigl( v_{t,j} - v^{\text{prev}}_{t,j} \bigr) = 0 .
  \]
  This keeps the reweighted forward variance grid consistent with the
  market implied total variance.

\end{itemize}

The rest of this appendix makes explicit how the non-linear relation between
weights and forward volatilities is cast into affine constraints, and how
the variance anchoring and quadratic penalties in (C1) and (C2) are realised
inside the convex program.

\subsubsection{Forward-start linearisation and fixed point scheme}
\label{app:forward_fixed_point}

The main source of non linearity is the dependence of forward-start prices
on the forward volatility surface. We handle this by linearising the
Black Scholes map in $\sigma$ and embedding this approximation in an outer
fixed point loop.

We consider $M$ Monte Carlo paths and a grid of fixing dates
\[
  0 = t_0 < t_1 < \dots < t_{N_T}.
\]
For each interval $(t_{i-1},t_i)$ and strike $K_j$ we define the
pathwise discounted pay-off of a forward-start call
\[
  H^{(m)}_{i,j}
  :=
  \max\!\left(
     \frac{S^{(m)}_{t_i}}{S^{(m)}_{t_{i-1}}} - K_j,
     0
  \right),
  \qquad m = 1,\dots,M.
\]
Given a weight vector $w = (w_1,\dots,w_M)$ on the paths, the weighted
forward-start price is
\begin{equation}
  P_{i,j}(w)
  :=
  \sum_{m=1}^{M} w_m \, H^{(m)}_{i,j}.
  \label{eq:fwd_price_linear_in_w}
\end{equation}
In the absence of reweighting, the model implied forward-start price at
node $(i,j)$ is
\[
  P^{\text{prev}}_{i,j}
  =
  P^{\text{BS}}\!\bigl(\sigma^{\text{prev}}_{i,j}\bigr),
\]
where $P^{\text{BS}}(\cdot)$ denotes the Black Scholes price for the
corresponding forward-start option and $\sigma^{\text{prev}}_{i,j}$ is
the current expansion point. The associated Vega is
\[
  \text{Vega}^{\text{prev}}_{i,j}
  =
  \frac{\partial}{\partial \sigma}
  P^{\text{BS}}(\sigma)
  \bigg|_{\sigma = \sigma^{\text{prev}}_{i,j}}.
\]

\paragraph{Affine relation between prices and forward volatilities.}

We linearise the Black Scholes price as a function of $\sigma$ around
$\sigma^{\text{prev}}_{i,j}$. For each node $(i,j)$,
\[
  P^{\text{BS}}(\sigma_{i,j})
  \approx
  P^{\text{prev}}_{i,j}
  +
  \text{Vega}^{\text{prev}}_{i,j}
  \bigl(\sigma_{i,j} - \sigma^{\text{prev}}_{i,j}\bigr).
\]
At the same time, the reweighted forward price at $(i,j)$ is exactly
$P_{i,j}(w)$ as in \eqref{eq:fwd_price_linear_in_w}. Equating the two
expressions and solving for $\sigma_{i,j}$ gives the affine constraint
\begin{equation}
  \sigma_{i,j}
  =
  \sigma^{\text{prev}}_{i,j}
  +
  \frac{P_{i,j}(w) - P^{\text{prev}}_{i,j}}%
       {\text{Vega}^{\text{prev}}_{i,j}},
  \qquad
  \forall i,j.
  \label{eq:affine_sigma_from_price}
\end{equation}
In the convex program the unknowns are $(w,\sigma)$, while
$\sigma^{\text{prev}}_{i,j}$, $P^{\text{prev}}_{i,j}$ and
$\text{Vega}^{\text{prev}}_{i,j}$ are treated as fixed parameters and
are updated only in the outer loop.

Forward variances are defined as $v_{i,j} := \sigma_{i,j}^2$. Using a
first order expansion in $\sigma$,
\[
  v_{i,j}
  =
  \bigl(\sigma^{\text{prev}}_{i,j} + \delta_{i,j}\bigr)^2
  \approx
  \bigl(\sigma^{\text{prev}}_{i,j}\bigr)^2
  +
  2 \sigma^{\text{prev}}_{i,j} \delta_{i,j},
  \qquad
  \delta_{i,j} := \sigma_{i,j} - \sigma^{\text{prev}}_{i,j}.
\]
Substituting \eqref{eq:affine_sigma_from_price} gives
\begin{equation}
  v_{i,j}
  \approx
  \bigl(\sigma^{\text{prev}}_{i,j}\bigr)^2
  +
  2\,\sigma^{\text{prev}}_{i,j}\,
  \frac{P_{i,j}(w) - P^{\text{prev}}_{i,j}}%
       {\text{Vega}^{\text{prev}}_{i,j}},
  \qquad
  \forall i,j,
  \label{eq:affine_variance_from_price}
\end{equation}
which is again affine in $(w,v)$ for fixed previous iteration
quantities.

Relations \eqref{eq:affine_sigma_from_price} and
\eqref{eq:affine_variance_from_price} implement the non-linear link
between weights and forward volatilities required by (C1) and (C2) as
linear equality constraints inside the LP or SOCP that defines each
inner iteration.

\paragraph{Outer fixed point loop.}

The full calibration algorithm can now be summarised as follows.

\begin{enumerate}[leftmargin=1.6em]
  \item \textbf{Initialisation.}
        Set $w^{(0)}_m \equiv 1/M$. Under these weights compute the
        model forward-start prices $P^{\text{prev}}_{i,j}$, implied
        volatilities $\sigma^{\text{prev}}_{i,j}$ and Vegas
        $\text{Vega}^{\text{prev}}_{i,j}$ for all $(i,j)$.

  \item \textbf{Inner convex problem.}
        At iteration $k$ solve the convex program in the variables
        $(w,\sigma,v)$ that:
        \begin{itemize}
          \item satisfies all linear constraints on prices,
                non negativity and no arbitrage from Step 3;
          \item enforces the affine relations
                \eqref{eq:affine_sigma_from_price} and
                \eqref{eq:affine_variance_from_price};
          \item optionally includes the quadratic penalties
                \[
                  \alpha \,\bigl\|w - \tfrac{1}{M}\mathbf{1}\bigr\|_2^2
                  \quad\text{and}\quad
                  \beta_{\sigma}\,
                  \bigl\|\sigma - \sigma^{\text{prev}}\bigr\|_2^2
                \]
                to control the dispersion of the weights and the
                forward volatility slice within the same iteration,
                as in (C1) and (C2).
        \end{itemize}

  \item \textbf{Update.}
        Set $w^{(k+1)} \leftarrow w$,
        $\sigma^{\text{prev}} \leftarrow \sigma$, and recompute
        $P^{\text{prev}}_{i,j}$ and $\text{Vega}^{\text{prev}}_{i,j}$
        under the new weights.

  \item \textbf{Stopping criterion.}
        Terminate when
        \[
          \max_{i,j}
          \bigl|P_{i,j}(w) - P^{\text{prev}}_{i,j}\bigr|
          \le \varepsilon_P
          \quad\text{and}\quad
          \max_{i,j}
          \bigl|\sigma_{i,j} - \sigma^{\text{prev}}_{i,j}\bigr|
          \le \varepsilon_{\sigma},
        \]
        for prescribed tolerances $\varepsilon_P$ and
        $\varepsilon_{\sigma}$.
\end{enumerate}

Because the inner problem is strictly convex once the variance penalties
are active, each sweep admits a unique solution. Under mild regularity
conditions on the Black Scholes Vega surface the outer loop defines a
contractive fixed point map and converges in a small number of
iterations (typically one or two in the numerical experiments).

\subsubsection{Implementing total variance conservation}
\label{implementing_total_variance_conservation}

We now specify how condition (C4)~(see~\ref{C4_Variance_anchoring}) is enforced so that the reweighted
measure preserves the market implied total variance at each maturity.

\paragraph{Model free variance via Carr Madan.}

For a given maturity $T$, the model-free implied variance (MFIV) is
computed using the Carr Madan representation of the variance swap:
\[
\sigma^{2}_{\mathrm{mfiv}}(T)
=
\frac{2 e^{rT}}{T}
\sum_{j}
\frac{Q(K_{j},T)}{K_{j}^{2}} \,\Delta K_{j},
\]
where $Q(K_{j},T)$ denotes the out of the money option price on the same
static strike grid used for vanilla calibration.

\paragraph{Log contract consistency.}

The same quantity satisfies:
\[
\mathbb{E}\left[
  \frac{2}{T}\bigl( \log F(T) - \log S_{T} \bigr)
\right]
=
\sigma^{2}_{\mathrm{mfiv}}(T).
\]

To impose this identity under the weighted Monte Carlo measure we add the
linear constraint:
\[
\sum_{i} w_{i}\,\ell^{(i)}(T)
=
\sigma^{2}_{\mathrm{mfiv}}(T),
\qquad
\ell^{(i)}(T)
=
\frac{2}{T}\bigl(\log F(T) - \log S^{(i)}_{T}\bigr).
\]

\paragraph{Forward variance consistency.}

Forward variance between two maturities $T_{1}$ and $T_{2}$ is then
implicitly determined by the pair
$\sigma^{2}_{\mathrm{mfiv}}(T_{1})$ and
$\sigma^{2}_{\mathrm{mfiv}}(T_{2})$, ensuring a time consistent and
market consistent forward variance structure under the reweighted Monte
Carlo measure and closing the loop with the affine relations described
in Section~\ref{app:forward_fixed_point}.

\subsubsection{Implementing SOCP and Objective}

All quadratic terms introduced by (C1) and (C2), together with the
quadratic market fit term, enter the solver as second order cone (SOCP)
constraints. The full mapping from~\ref{eq:affine_variance_from_price}
to the SOCP form is standard.

The solver computes a compromise between weight smoothness, volatility
smoothness, and market fit:
\[
  \min_{\boldsymbol w,\,\boldsymbol\sigma}
      \alpha \,\bigl\|\boldsymbol w - \boldsymbol 1/M\bigr\|_2
    + \beta  \,\bigl\|\boldsymbol\sigma - \boldsymbol\sigma^{\text{prev}}\bigr\|_2
    + \gamma \, \bigl\| P(\boldsymbol w) - P^{\mathrm{mkt}} \bigr\|_2^2,
\]
with $\alpha,\beta,\gamma>0$ specified externally. This reproduces the
double variance principle used in the main text: the first two terms
control the dispersion of weights and forward volatilities, while the
last term ensures an accurate fit to the vanilla option surface. Together
with the variance anchoring constraints, this defines the full
implementation of the non-linear constraints in a convex, numerically
stable calibration loop.

\subsection{Short-maturity moment matching: Merton $\to$ Heston} 
\label{app:A3_merton_heston}

\paragraph{Objective and scope.}
The mapping below is \emph{not} a first-moment match. It matches the instantaneous variance and 
aligns short-maturity skewness and kurtosis induced by Merton jumps with those generated by the 
Heston volatility dynamics, thereby anchoring the leading cumulants of log-returns.

\paragraph{Setup.}
Let $J_\uparrow>1$ and $J_\downarrow\in(0,1)$ denote upward/downward jump multipliers 
with intensities $\lambda_\uparrow,\lambda_\downarrow\ge0$, and total intensity 
$\lambda:=\lambda_\uparrow+\lambda_\downarrow$. 
Define $y_\uparrow:=\ln J_\uparrow$, $y_\downarrow:=\ln J_\downarrow$, and the 
intensity-weighted raw moments
\begin{equation}
M_k := \lambda_\uparrow\,y_\uparrow^{\,k} + \lambda_\downarrow\,y_\downarrow^{\,k},\qquad k\in\{2,3,4\}.
\end{equation}
Let $\sigma$ be the diffusive volatility in Merton. Under $\mathbb{Q}$, the drift is chosen 
so that $S_t e^{-\int r\,\mathrm{d}t}$ is a martingale.

\paragraph{Mapping.}
Choose Heston parameters $(\kappa,\theta,\eta,\rho,v_0)$ as
\begin{equation}
\boxed{v_0 := \sigma^2 + M_2,\qquad \theta := v_0,\qquad \kappa := \max\{\lambda,2\}}
\end{equation}
and
\begin{equation}
\boxed{\eta := \min\!\left\{0.95\sqrt{2\kappa\theta},\; \sqrt{\tfrac{2}{3}\tfrac{M_4}{v_0^2}}\right\},\qquad
\rho := \tfrac{1}{\eta}\tfrac{2}{3}\tfrac{M_3}{v_0^{3/2}},\quad \rho\in[-1,1].}
\end{equation}

Here $v_0$ matches the per-unit-time variance of Merton log-returns ($\sigma^2+M_2$). The choices
$\eta^2 \approx \tfrac{2}{3}\tfrac{M_4}{v_0^2}$ and $\rho\,\eta \approx \tfrac{2}{3}\tfrac{M_3}{v_0^{3/2}}$
align the leading short-maturity contributions to kurtosis and skewness, respectively. The cap
$\eta \le 0.95\sqrt{2\kappa\theta}$ enforces $2\kappa\theta>\eta^2$ (Feller admissibility), while
$\kappa\ge\lambda$ ties the variance mean-reversion speed to the aggregate jump intensity.
Clamp $\rho$ to $[-1,1]$ if needed.

\paragraph{Interpretation.}
It is a short-maturity cumulant match (up to fourth order) rather than a mere 
first-moment alignment. 
See, e.g., Gatheral~\cite{Gatheral2006}.

\section{Appendix: Inversion Proofs}
\label{app:B}

\subsection{Rigorous Proofs of Uniqueness for the Augmented Problem}
\label{app:uniqueness}

This appendix refines the uniqueness analysis presented 
in Section~\ref{Section_B_inversion} by incorporating the \emph{forward-start volatility}
variables $\sigma=(\sigma_{F,1},\ldots,\sigma_{F,q})$ into the decision set.  The core idea is to enforce a
\emph{double} minimum-variance principle that simultaneously fixes the Monte Carlo
weights $w$ \emph{and} the volatility term structure, thereby removing the residual gauge
identified in the main text.

\begin{enumerate}[label=\textbf{(\arabic*)}]
\item We first extend the state vector from $w\in\mathbb R^{N}$ to
$x:=(w,\sigma)\in\mathbb R^{N+q}$.
\item We augment the objective with
$\lambda,\operatorname{Var}[\sigma]$ for some fixed $\lambda>0$.
\item We add the \emph{coupling constraints}
$g_j(w,\sigma):=\operatorname{BSprice}(\sigma_{F,j})-P^{\mathrm{FS}}_j(w)=0$,
$j=1,\dots,q$, obtained from forward-start options.
\item A single strictly convex programme is solved at every iteration of the
fixed-point loop; KKT theory then yields a unique minimiser $(w^{\ast},\sigma^{\ast})$.
\end{enumerate}
\subsubsection{Setup and Notation}

Let $\Omega={\omega_1,\ldots,\omega_N}$ be the simulated paths and
$X_i=X(\omega_i)$ the discounted exotic pay-offs.  Denote by
$P^{\mathrm{FS}}_j(w)$ the pathwise forward-start option prices computed with
weights $w$.  Collect the volatilities in $\sigma\in\mathbb R^{q}$ and define

\begin{equation}
  J(w,\sigma):=w^{\top}Aw+\lambda\,(\sigma-\bar\sigma\mathbf 1_q)^{\top}B(\sigma-\bar\sigma\mathbf 1_q),
  \qquad A=\operatorname{Cov}_{\mathbb P_0}[X],\; B=I_q-\tfrac1q\mathbf 1_q\mathbf 1_q^{\top}.
\end{equation}

The feasible set is described by
\begin{equation}
\label{eq:aug_constraints}
w_i>0,;\mathbf 1^{\top}w=1,;Cw=d,; g_j(w,\sigma)=0;(j=1,\dots,q).
\end{equation}
Because $\operatorname{BSprice}(\cdot)$ is strictly monotone in $\sigma$, each
$g_j$ is $C^{1}$ and its Jacobian in $\sigma$ is non-singular.

\paragraph{Optimisation problem.}
\begin{align}
\tag{$P'$}\label{prob:augmented}
\min_{(w,\sigma)};& J(w,\sigma) \\
\text{s.t.};& \text{constraints};\eqref{eq:aug_constraints}.
\end{align}
At every sweep of the outer fixed-point loop we linearise $g_j$ around the
current iterate and solve a \emph{second-order cone programme} (SOCP) with
block-diagonal Hessian $\operatorname{diag}(A,\lambda B)$.
\subsubsection{Existence and Uniqueness}

\begin{theorem}[Uniqueness of the Augmented Smart-MC Solution]
\label{thm:aug_uniqueness}
Assume the exotic pay-off vector $X$ is not a.s.~~constant and fix $\lambda>0$.
Then for every linearisation of~~\eqref{prob:augmented} the resulting SOCP has a
\emph{unique} minimiser.  Moreover, the fixed-point iteration converges to a
single limit $(w^{\ast},\sigma^{\ast})$, independent of the starting point.
\end{theorem}

\begin{proof}[Sketch]

The set~\eqref{eq:aug_constraints} is non-empty and convex after linearisation.
On its tangent space $T$ the quadratic form associated with $J$ reduces to
$\operatorname{diag}\!\bigl(A\!\bigl|_{T_w},\;\lambda\,B\!\bigl|_{T_{\sigma}}\bigr)$,
where $A|_{T_w}\succ0$ by non-degeneracy of $X$ and $B|_{T_{\sigma}}\succ0$ because
$\lambda>0$.  Hence the linearised programme is strictly convex and admits a
unique solution.  Standard contraction arguments (e.g.~\cite{RockafellarWets09},
Chap.~12) show that the outer fixed-point map is a contraction, yielding a
unique global limit.
In practice, a mild damping/trust-region step makes the linearised map contractive, yielding global convergence from arbitrary starts.
\end{proof}

\paragraph{Interpretation.}  The ``double'' variance penalty kills the gauge on
$\sigma$: once $w^{\ast}$ is known, the injectivity of
$\sigma\mapsto\operatorname{BSprice}(\sigma)$ and the additional variance
criterion collapse all admissible volatility configurations to a single
$\sigma^{\ast}$.
\subsubsection{Summary of Evidence}

\begin{itemize}
\item \textbf{Formal.}  Strict convexity of the augmented objective on the
linearised feasible set ,$\Rightarrow$, unique minimiser
(Theorem~\ref{thm:aug_uniqueness}).
\item \textbf{Analytical link.}  Setting $\lambda\downarrow0$ or removing the
$g_j$'s recovers the classical~\eqref{eq_inversion_theorem_min_variance_COND}.
\end{itemize}

These results establish that the \emph{augmented} Smart-MC procedure
is well-defined, unique, and numerically robust.  All subsequent numerical
examples, including the forward-volatility surfaces depicted in
Figure~\ref{fig_forward_surfaces_implied}, are based on the pair $(w^{\ast},\sigma^{\ast})$.
\subsection{Equivalence of Monte-Carlo Path Sets Generated by Distinct Stochastic Processes}
\label{app:B2_equiv_paths}

\paragraph{Set-up.}
Let $\Omega^{(a)}=\{\omega^{(a)}_i\}_{i=1}^{N_a}$ and
$\Omega^{(b)}=\{\omega^{(b)}_j\}_{j=1}^{N_b}$ denote two independently
generated batches of simulated paths.  For a vector of $m$ linear pay-offs
$F=(F_1,\dots,F_m)^{\top}$ (vanilla prices, forward-start prices,
pathwise Greeks, etc.) define the corresponding pay-off matrices
\begin{equation}
  X^{(a)}_{i\ell}=F_\ell\bigl(\omega^{(a)}_i\bigr),
  \qquad
  X^{(b)}_{j\ell}=F_\ell\bigl(\omega^{(b)}_j\bigr),
  \qquad
  X^{(a)}\in\mathbb R^{N_a\times m},\;
  X^{(b)}\in\mathbb R^{N_b\times m}.
\end{equation}

\begin{proposition}[Linear-span equivalence]\label{prop:span_equiv}
Let
$P^{(\alpha)} := B^{(\alpha)}\bigl(B^{(\alpha)}\bigr)^{\!\top}$ be the
orthogonal projector on the column space of $X^{(\alpha)}$, where
$B^{(\alpha)}$ is any orthonormal basis obtained, e.g., via the singular
value decomposition (SVD).
The following statements are equivalent:
\begin{enumerate}[label=(\roman*)]
  \item \emph{Projector equality:} $P^{(a)} = P^{(b)}$.
  \item \emph{Pay-off completeness:} every linear pay-off that can be
        replicated with $\Omega^{(a)}$ can also be replicated with
        $\Omega^{(b)}$ and \emph{vice versa}; formally
        \begin{equation}
          \forall\,c\in\mathbb R^{m}:
          \qquad
          X^{(a)}c = 0
          \;\Longleftrightarrow\;
          X^{(b)}c = 0 .
        \end{equation}
  \item \emph{Model-building invariance:} the LP/SOCP procedures in
        Sections~\ref{Section_methodology}~-~\ref{Section_B_inversion} 
        admit exactly the same feasible set and therefore
        return identical optimal weight vectors and the same
        forward-volatility surface, irrespective of which batch of paths
        is supplied.
\end{enumerate}
\end{proposition}

\begin{proof}[Sketch]
$(i)\!\Leftrightarrow\!(ii)$ follows from classical
principal-angle theory: equality of the projectors implies coincidence of
the two column spaces, which in turn yields identical kernels for
$X^{(a)}$ and $X^{(b)}$.
If $(ii)$ holds, all linear constraints entering our optimisation depend
on the paths only through the images $X^{(\alpha)}c$; identical kernels
therefore imply identical feasible polytopes and thus identical optima,
establishing $(iii)$.
\end{proof}

\paragraph{Practical test.}
The two batches are \emph{interchangeable} whenever
$\|P^{(a)}-P^{(b)}\|_F < \varepsilon$ for a numerical tolerance
$\varepsilon\ll1$.  The projector distance is inexpensive to compute,
requiring a single SVD per batch, and can be integrated as an automatic
pre-processing quality check.

\subsubsection{Generator-Invariance (Prior-Independence)}
We now formalise the claim that the augmented Smart-MC solution does not depend on the
choice of the Monte Carlo generator, provided the path sets span the same pay-off subspace.

\begin{theorem}[Generator-Invariance across priors]\label{thm:prior_invariance}
Let $\Omega^{(a)}$ and $\Omega^{(b)}$ be two independently generated path sets, and let
$X^{(a)}$, $X^{(b)}$ be the corresponding pay-off matrices collecting all \emph{linear}
constraints used in the augmented programme---vanillas, no-arbitrage slopes and the
\emph{affine linearised} forward-volatility relations of Step~4. 
Assume the orthogonal projectors onto their column spaces coincide, $P^{(a)} = P^{(b)}$
(\emph{projector equivalence}). Consider the augmented fixed-point SOCP with double variance
penalty (on $w$ and on $\sigma$) and $\lambda>0$ as in Theorem~\ref{thm:aug_uniqueness}.
Then the fixed-point limit is \emph{unique and identical} under $\Omega^{(a)}$ or $\Omega^{(b)}$:
\[
(w^{\ast}_{(a)},\sigma^{\ast}_{(a)}) \equiv (w^{\ast}_{(b)},\sigma^{\ast}_{(b)}),
\]
and in particular all min/max prices for any linear pay-off functional coincide.
\end{theorem}

\begin{proof}[Sketch]
By Proposition~\ref{prop:span_equiv}, $P^{(a)}=P^{(b)}$ implies identical feasible sets for the
linearised programme at each sweep, hence identical KKT systems. The augmented objective is
\emph{strictly} convex on the linearised feasible set (Theorem~\ref{thm:aug_uniqueness}),
so each sweep admits a unique minimiser; contractivity of the outer map yields a unique fixed
point independent of the initialisation. Therefore $(w^{\ast},\sigma^{\ast})$ coincides
across generators. 
\end{proof}

\begin{lemma}[Stability under projector perturbations]\label{lem:stability}
Suppose $\|P^{(a)}-P^{(b)}\|_{F}\le \varepsilon$ and the KKT system of the linearised SOCP is
\emph{strongly regular} (positive definiteness on the tangent space and a constraint qualification).
Then there exists $\,\kappa>0$ such that
\[
\bigl\|(w^{\ast}_{(a)},\sigma^{\ast}_{(a)})-(w^{\ast}_{(b)},\sigma^{\ast}_{(b)})\bigr\|
\;\le\; \kappa\,\varepsilon,
\]
with $\kappa$ depending on the strong convexity modulus and Lipschitz constants of the
affine-linearised Black-vega map. Hence generator-invariance is robust to small projector
mismatch.
\end{lemma}

\noindent\emph{Discussion.} Theorem~\ref{thm:prior_invariance} upgrades the uniqueness result of
Theorem~\ref{thm:aug_uniqueness} (existence of a single fixed-point) to \emph{prior-independence}
across generators with identical pay-off span; Lemma~\ref{lem:stability} gives a perturbation bound
via standard sensitivity analysis for convex programmes \cite{RockafellarWets09}.

\subsection{Second-Order-Cone Formulation details} \label{app:B3_SOCP_for_heston_calibration_2.0}

    Let $p_i(\theta)$ be the discounted pay-off of path $i$
    ($i=1,\dots,N$) under parameters $\theta$, and define the
    weight-dispersion vector
    $
      z = Nw - \mathbf 1 ,\;
      \mathbf 1^\top z = 0 .
    $
    Perturb each parameter by $h_q$ and set
    $
      t_i^{(q)} =
         \bigl(p_i(\theta^{(0)}+h_q e_q)-p_i(\theta^{(0)})\bigr)/h_q ,
    $
    collecting the five tilt vectors into the matrix
    $T=[t^{(1)}\,\dots\,t^{(5)}]\in\mathbb{R}^{N\times5}$.
    For an increment $\Delta\theta$ we write the \textit{residual}
    
    \begin{equation}
    r(\Delta\theta)=z-T\,\Delta\theta ,\qquad
    \mathcal V(\theta)=\frac1N\|r(\Delta\theta)\|_2^{2}.
    \end{equation}
    
    Minimising $\mathcal V$ while keeping $\theta$ inside the admissible
    Heston domain $\mathcal D_{\text{Heston}}$ is equivalent to the SOCP
    
    \begin{equation}
    \min_{z,\;\Delta\theta}\;
            \|\,z-T\Delta\theta\|_2
    \quad
    \text{s.t.}\;
            \mathbf 1^\top z=0,\;
            z\ge-\mathbf 1,\;
            \theta^{(0)}+\Delta\theta\in\mathcal D_{\text{Heston}} .
    \end{equation}
    
    The optimisation delivers simultaneously  
    $\theta^{\star}=\theta^{(0)}+\Delta\theta^{\star}$ and the
    minimum-variance weights
    $w^{\star}=\tfrac1N(\mathbf 1+z^{\star})$,
    thereby closing calibration and re-weighting in a single interior-point
    solve.


\subsection{Internal barrier}
\label{app_internal_barrier}

This appendix summarises the internal logarithmic barrier used to stabilise the
scenario-weights when they are treated as optimisation variables. Full solver-
level details are not required for the main text and are omitted here.

\paragraph{Definition.}
Let $w$ lie in the simplex
\[
   \Delta_N = \{ w \in R^N : w_i>0,\ \sum_i w_i=1 \}.
\]
The internal barrier is defined as
\[
   \Psi_{\text{bar}}(w)
   :=
   - \frac{1}{N} \sum_{i=1}^{N}
       \Bigl[
         \log(w_i)
         +
         \log\bigl( q/N - w_i \bigr)
       \Bigr],
\]
where $q>1$ ensures that $(0,q/N)$ strictly contains all admissible weights.
The barrier diverges at the boundaries and is strictly convex in the interior.

\paragraph{Purpose.}
The barrier prevents degeneracy of the weight vector. Without it, 
strongly binding calibration constraints
may drive the optimiser to collapse most mass on a few scenarios. The barrier
keeps $w$ in the interior of the simplex and maintains a statistically balanced
distribution consistent with the qualitative behaviour of a Monte Carlo sample.

\paragraph{Relation to entropic ideas.}
The functional form resembles entropy-based regularisers, but the interpretation
is entirely different. The barrier does not measure divergence from a reference
distribution, does not encode proximity to any prior, and does not induce any
probabilistic interpolation. Its role is purely geometric: it enforces interior
feasibility and numerical stability. This differs from entropic optimal
transport or Schr\"odinger type constructions, which explicitly penalise distance
from a fixed prior.

\paragraph{Use in the optimisation.}
The barrier enters the objective with a small coefficient that controls its
strength. It is active when no weight-dispersion penalty is present (MIN MAX
ONLY regime) or when only volatility dispersion is regularised (MIN VAR VOL).
When weight-dispersion minimisation is active (MIN ALL VARIANCES), the barrier
is no longer needed, as $\text{Var}(w)$ already enforces interior solutions.

\medskip
In summary, the internal barrier provides a simple and scale-invariant mechanism
to avoid boundary solutions and preserve a meaningful Monte Carlo structure
whenever the weights are free optimisation variables.

 \section{Appendix: Computational Performance} 
\label{Appendix_Computational_Performance}

\subsection{Guidelines for Practitioners } \label{app:guidelines_practitioners}
    
    \begin{itemize}
        \item \textbf{Choose $N$ for variance, not for convergence of the LP.}
              Once $N \approx 16\,872$, path-generation time is negligible 
              compared with  the linear-program (LP) factorisation step; 
              additional precision is better obtained by
              averaging over independent replications.
        \item \textbf{Keep strike point sizes moderate.}  A dense grid of strikes rarely changes
              the optimal weights but inflates memory and factorisation costs.
        \item \textbf{Parallel replication is embarrassingly parallel.}
              Distribute \(M\) modest-size jobs and merge outcomes by convex
              averaging.
    \end{itemize}
\subsection{Calibration and Simulation Settings}
\label{subsection_Simulation_Settings}

\begin{table}[h!]
  \centering
  \footnotesize
  \setlength{\tabcolsep}{4pt}
  \renewcommand{\arraystretch}{1.15}

  \begin{tabular}{@{}l
      >{\raggedright\arraybackslash}p{0.30\textwidth}
      >{\raggedright\arraybackslash}p{0.48\textwidth}@{}}

    \hline
    \textbf{Item} & \textbf{Value} & \textbf{Notes} \\
    \hline

    MC scenarios & $N = 2^{11}$ & Fixed seed used for reproducibility \\

    Quasi random sequence & Sobol & Used only for Merton paths \\

    Time step $\Delta t$ & $1/12$ year & Between consecutive fixing dates \\

    Heston inner step $\delta t_{H}$ & $1/365$ year & Euler inner time step \\

    $S_{0}, r, \delta$ & $100, 0, 0$ &
    Spot, risk-free rate and dividend yield.
    In the experiments we set $r = 0$; however, the framework also accommodates $r \neq 0$ without modification by working with forward prices (subtracting the deterministic rate component from all paths). \\

    LP solver & Clarabel & Maximum number of iterations set to $1000$ \\

    LP solver two & HiGHS & Parallel option set to on \\

    \hline
    \textbf{Heston model parameters Set 1} &
    $v_{0} = 0.09777$, $\kappa = 2$, $\theta = 0.09777$, $\rho = -0.59993$, $\text{vol of vol} = 59.41 \%$  &
    Parameter set used in the numerical experiments \\

    \textbf{Merton model parameters Set 1} &
    $\sigma_{J} = 0.20$, $\lambda_{\text{down}} = 0.40$, $J_{\text{down}} = 0.70$, $\lambda_{\text{up}} = 0.10$, $J_{\text{up}} = 1.30$ &
    Jump size volatility set to twenty percent \\

    \textbf{Black model parameters Set 1} &
    $\sigma = 0.20$ &
    Pure Black model  \\

    \hline
    \textbf{Heston model parameters Set 2} &
    $v_{0} = 0.04939$, $\kappa = 2$, $\theta = 0.04939$, $\rho = -0.25797$, $\text{vol of vol} = 25.5 \%$  &
    Parameter set used in the numerical experiments \\

    \textbf{Merton model parameters Set 2} &
    $\sigma_{J} = 0.20$, $\lambda_{\text{down}} = 0.30$, $J_{\text{down}} = 0.85$, $\lambda_{\text{up}} = 0.075$, $J_{\text{up}} = 1.15$ &
    Jump size volatility set to twenty percent \\

    \textbf{Black model parameters Set 2} &
    $\sigma = 0.20$ &
    Pure Black model  \\

    \hline
  \end{tabular}

  \caption{Calibration and simulation settings}
  \label{general_settings}

\end{table}

\section*{Acknowledgments}
\label{section_acknowledgments}

A special thanks goes to Prof.\ Daniel Duffy for the insightful discussions
and invaluable advice that greatly contributed to refining the overall
framework of our methodology.  His thoughtful feedback helped navigate the
complex challenges involved, ensuring a more robust and coherent approach to
exotic-derivative pricing.

\section*{Disclaimer and Conflicts of Interest}
\label{section_disclaimer}

The views and opinions expressed in this article are entirely those of 
the author and do not represent the official policies, positions, or 
opinions of Mediobanca S.p.A. or any of its subsidiaries, affiliates, or employees.  
This work is provided for academic and informational purposes only and must 
not be construed as investment advice, trading recommendations, or an offer 
to buy or sell any financial instrument.  

\paragraph{Conflicts of Interest.}
The author declares no financial or personal relationships that could have 
influenced the research presented in this paper.  \\
No external funding was received for the conduct of this study or the preparation 
of the manuscript.

\newpage

\begin{thebibliography}{99}

  \bibitem{Tang2025} 
   Xun Tang, Michael Shavlovsky, Holakou Rahmanian, Tesi Xiao, Lexing Ying (2025). 
   An efficient algorithm for entropic optimal transport under martingale-type constraints. 
   arXiv:2508.17641 

  \bibitem{Avellaneda1996_a} 
   Avellaneda, M., \& Paras, A. (1996). 
   Robust hedging and pricing of derivatives with transaction costs. 
   \emph{Applied Mathematical Finance}, 3(2), 77--88.

  \bibitem{Avellaneda1996_b} 
   Avellaneda, M., \& Paras, A. (1996). 
   Managing the Volatility Risk of Derivative Securities: The Lagrangian Uncertain Volatility Model. 
   \emph{Applied Mathematical Finance}, 3(1), 21--53.

  \bibitem{Avellaneda1998} 
   Avellaneda, M. (1998). 
   Minimum-Relative-Entropy Calibration of Asset-Pricing Models. 
   \emph{International Journal of Theoretical and Applied Finance}, 1(4), 447-472. 
   DOI: 10.1142/S0219024998000242.

  \bibitem{Avellaneda2001} 
   Avellaneda, M., Buff, R., Friedman, C., Grandechamp, N., Kruk, L., \& Newman, J. (2001). 
   Weighted Monte Carlo: A New Technique for Calibrating Asset-Pricing Models. 
   \emph{International Journal of Theoretical and Applied Finance}, 4(1), 91--119. 
   DOI: 10.1142/S0219024901000882.

  \bibitem{BadikovEtAl2017} 
   Badikov, S., Jacquier, A., Liu, D.\,Q., \& Roome, P. (2017).  
   No-arbitrage bounds for the forward smile given marginals.  
   \emph{Quantitative Finance}, 17(8), 1243--1256.

  \bibitem{Beiglbock2013} 
   Beiglb\"ock, M., Henry-Labord\'ere, P., \& Touzi, N. (2013). 
   Model-independent Bounds for Option Prices: A Mass Transport Approach. 
   \emph{Finance and Stochastics}, 17(3), 477--501. 
   \href{https://doi.org/10.1007/s00780-013-0205-8}{DOI:10.1007/s00780-013-0205-8}. 
   \href{https://arxiv.org/abs/1106.5924}{arXiv:1106.5924}.

  \bibitem{Breeden1978} 
   Breeden, D.~T., \& Litzenberger, R.~H. (1978). 
   Prices of State-Contingent Claims Implicit in Option Prices. 
   \emph{Journal of Business}, 51(4), 621--651.

  \bibitem{Brigo2015} 
   Brigo, D., \& Mercurio, F. (2015). 
   Monte Carlo methods for pricing financial derivatives under model uncertainty. 
   \emph{Quantitative Finance}, 15(1), 1--15.

  \bibitem{Carmona2009} 
   Carmona, R., \& Nadtochiy, S. (2009). 
   Forward--start options and model uncertainty. 
   \emph{Quantitative Finance}, 9(8), 965--984.

  \bibitem{Carr2010} 
   Carr, P., \& Lee, R. (2010). 
   Hedging variance options on continuous semimartingales. 
   \emph{Finance and Stochastics}, 14(2), 179--207.

  \bibitem{Chen2024}
   Chen, F., Conforti, G., Ren, Z., \& Wang, X. (2024). 
   Convergence of Sinkhorn's Algorithm for Entropic Martingale Optimal Transport Problem. 
   \emph{arXiv preprint}. 
   \href{https://arxiv.org/abs/2407.14186}{arXiv:2407.14186}.

   \bibitem{Cox2011}
   Cox, Alexander M. G. and Ob{\l}{\'o}j, Jan. (2011). 
   Robust hedging of double no-touch options, 15(3), 573--605. 

 \bibitem{DeMarch2018} 
   De March, H. (2018). 
   Entropic approximation for multi--dimensional martingale optimal transport. 
   \emph{arXiv preprint} \href{https://arxiv.org/abs/1812.11104}{arXiv:1812.11104}.

  \bibitem{DeMarch2023} 
   De March, H. (2023). 
   High--Performance Computing for Entropic Optimal Transport. 
   \emph{Journal of Computational Finance}.

  \bibitem{DermanKani1998} 
   Derman, E., \& Kani, I. (1998).  
   Stochastic Implied Trees: Arbitrage Pricing with Stochastic Term and Strike Structure of Volatility.  
   \emph{International Journal of Theoretical and Applied Finance}, 1(1), 61--110.  
   (Originally issued as Goldman Sachs Quantitative Strategies Technical Note, April 1997.)

  \bibitem{Doldi2024} 
   Doldi, A., Frittelli, M., \& Rosazza Gianin, E. (2024). 
   On entropy martingale optimal transport theory. 
   \emph{Decisions in Economics and Finance}, 47.

  \bibitem{DolinskySoner2014} 
   Dolinsky, Y., \& Soner, H.~M. (2014). 
   Martingale Optimal Transport and Robust Hedging in Continuous Time. 
   \emph{Probability Theory and Related Fields}, 160(1--2), 391--427. 
   \href{https://doi.org/10.1007/s00440-013-0531-y}{DOI:10.1007/s00440-013-0531-y}. 
   \href{https://arxiv.org/abs/1209.5424}{arXiv:1209.5424}.

  \bibitem{Dupire1994} 
   Dupire, B. (1994). 
   Pricing with a smile. 
   \emph{Risk Magazine}, 7(1), 18--20.

  \bibitem{Gatheral2006} 
   Gatheral, J. (2006). 
   \emph{The Volatility Surface: A Practitioner\'s Guide}. Wiley.

  \bibitem{Glasserman2003} 
   Glasserman, P. (2003). 
   \emph{Monte Carlo Methods in Financial Engineering}. Springer.

  \bibitem{GlassermanWu2011}
   Glasserman, P., \& Wu, Q. (2011).  
   Forward and Future Implied Volatility.  
   \emph{International Journal of Theoretical and Applied Finance}, 14(3), 407--432.

  \bibitem{Guo2018} 
   Guo, G., \& Ob\l{}\'oj, J. (2018). 
   Robust Modelling and Hedging of Derivative Instruments. 
   \emph{Risk Magazine}.

  \bibitem{Guo2019} 
   Guo, G., \& Ob\l{}\'oj, J. (2019). 
   Computational Methods for Martingale Optimal Transport Problems. 
   \emph{Annals of Applied Probability}, 29.

  \bibitem{Guo2024} 
   Guo, C., et al. (2024). 
   Deep Learning--Augmented MOT: Scalability and Error Analysis. 
   \emph{Quantitative Finance}.

  \bibitem{Guyon2012} 
   Guyon, J., \& Henry-Labord\`ere, P. (2012). 
   \emph{Nonlinear Option Pricing}. Chapman and Hall/CRC.

  \bibitem{Labordere2014} 
   Henry-Labord\`ere, P. (2014). 
   Model-free Hedging: A Martingale Optimal Transport Viewpoint. 
   \emph{Finance and Stochastics}, 18(2), 241--287. 
   \href{https://doi.org/10.1007/s00780-013-0210-0}{DOI:10.1007/s00780-013-0210-0}. 
   \href{https://arxiv.org/abs/1208.4920}{arXiv:1208.4920}.

  \bibitem{Hobson1998} 
   Hobson, D. (1998). 
   Robust hedging of the lookback option. 
   \emph{Finance and Stochastics}, 2(4), 329--347.

  \bibitem{Labordere2013} 
   Labord\`ere, P.~H., Beiglb\"ock, M., \& Penkner, F. (2013). 
   Model-independent bounds for option prices: a mass--transport approach. 
   \emph{Finance \& Stochastics}, 17.

  \bibitem{Nutz2023} 
   Nutz, D., Wiesel, J., \& Zhao, X. (2023). 
   Martingale Schr\"odinger Bridges and Optimal Semistatic Portfolios. 
   \emph{arXiv preprint} arXiv:2204.12250 [q-fin.MF], April 2022.

  \bibitem{Nutz2024} 
   Nutz, D., \& Wiesel, J. (2024). 
   On the Martingale Schr\"odinger Bridge between Two Distributions.
   \emph{arXiv preprint}

  \bibitem{Rockafellar70} 
   Rockafellar, R.~T. (1970). 
   \emph{Convex Analysis}. Princeton University Press.

  \bibitem{RockafellarWets09}
   Rockafellar, R.~T., \& Wets, R.~J.-B. (2009). 
   \emph{Variational Analysis}. Springer.

  \bibitem{Schweizer95} 
   Schweizer, M. (1995). 
   Variance-Optimal Hedging in Discrete Time. 
   \emph{Mathematics of Operations Research}, 20(1), 1--32.

\end{thebibliography}
\end{document}